\documentclass[10pt]{article} 
\usepackage[preprint]{tmlr}

\usepackage{amsmath,amssymb,amsthm}
\usepackage{booktabs}
\usepackage{array}
\usepackage{longtable}
\usepackage{graphicx}
\usepackage[bookmarks=false]{hyperref}
\usepackage{url}


\theoremstyle{plain}
\newtheorem{proposition}{Proposition}
\newtheorem{theorem}[proposition]{Theorem}
\newtheorem{corollary}[proposition]{Corollary}
\newtheorem{lemma}[proposition]{Lemma}
\theoremstyle{definition}

\theoremstyle{remark}
\newtheorem{remark}[proposition]{Remark}

\newcommand{\R}{\mathbb{R}}
\newcommand{\E}{\mathbb{E}}
\newcommand{\Var}{\operatorname{Var}}
\newcommand{\Cov}{\operatorname{Cov}}
\newcommand{\tr}{\operatorname{tr}}
\newcommand{\rank}{\operatorname{rank}}
\newcommand{\spn}{\operatorname{span}}
\newcommand{\range}{\operatorname{range}}
\newcommand{\Nor}{\mathcal{N}}
\newcommand{\PhiN}{\Phi_{\mathcal{N}}}
\newcommand{\Pperp}{P_{\perp}}
\newcommand{\Ppar}{P_{\parallel}}
\newcommand{\Rstd}{\widetilde{R}}
\newcommand{\Upop}{U_{\mathrm{pop}}}

\begin{document}

\title{When Noise Estimation Hides Basis Misspecification in\\
Repeated Bayesian Inverse Problems}
\author{\name Rares Dimitrie Grozavescu \email rg625@cam.ac.uk \\
        \addr University of Cambridge
        \AND
        \name Mark Girolami \\
        \addr University of Cambridge}
\maketitle

\begin{abstract}
Basis-restricted priors in Bayesian inverse problems can lose coverage when the truth has components outside the basis. We show that estimating the observation-noise variance can hide this loss. Under a linear forward model, when the in-span prior variance dominates the noise, the maximum-likelihood noise estimate absorbs the out-of-basis energy in the complement of the model range. Residual-magnitude and observation-coverage checks then stay near nominal while field coverage falls. We study repeated problems sharing one forward operator and one basis, fixed independently of the tested data. After projection onto the complement, and conditionally on the fitted noise scale, every exact test is a test of the scale-free direction of the residuals. We test the shape of their sample spectrum with John's sphericity statistic. Under Gaussian noise its null model is exact at finite sample size, and we derive its null mean and its power at proportional dimension. On synthetic problems and in a preregistered GEBCO topography study, the test detects structured out-of-basis variation that cross-validation and observation-coverage checks largely miss. It cannot detect Gaussian out-of-basis variation that is isotropic in the complement, since that is indistinguishable from a change of noise scale.
\end{abstract}

\section{Introduction}

Physical inverse problems recover a field $a$ from indirect, noisy observations
$y=Ha+\varepsilon$, where $H$ is a known forward operator and $\varepsilon$ is
Gaussian noise of scale $\sigma$. A common recipe restricts the field prior to a
fixed low-dimensional basis $\Phi$: a truncated spectral expansion, a
reduced-order basis, a set of physical modes, or the output range of a trained
decoder. The noise scale is fitted alongside the field, and calibration is
checked on held-out observations.

When the true field has components outside $\spn(\Phi)$, the restricted prior
oversmooths. Oversmoothing priors give credible sets whose frequentist coverage
degenerates, and rougher priors give conservative sets of the correct order
\citep{knapik2011bayesian,szabo2015frequentist}.

Write the true field as $a=\Phi w^\star+\xi_\perp$ with $\xi_\perp\perp\spn(\Phi)$. The
observations are then
\begin{equation*}
  y = H\Phi w^\star + H\xi_\perp + \varepsilon .
\end{equation*}
Let $\Pperp$ project onto the orthogonal complement of $\range(H\Phi)$, which
has dimension $n-r$. The projection removes the modelled term:
\begin{equation*}
  \Pperp y = \Pperp H\xi_\perp + \Pperp\varepsilon .
\end{equation*}
Orthogonal to $\range(H\Phi)$ the model has no structure, so $\sigma$ is the only
parameter left to explain the residual there. When the in-span prior scale
dominates the noise, the maximum-likelihood estimate is
\begin{equation*}
  \hat\sigma^2 \approx \|\Pperp y\|^2/(n-r),
\end{equation*}
so it absorbs the out-of-basis field energy (Proposition~\ref{prop:absorption}).
Observation coverage then stays near nominal. Across our sweep it
moved by $0.02$ while field coverage fell by more than $0.5$
(Figure~\ref{fig:sweep}, Appendix~\ref{app:sweep}); Figure~\ref{fig:twospaces}
shows one fit. Checks that read residual magnitude, such as
held-out observation coverage, therefore pass.

\begin{figure}[!t]
\centering
\includegraphics[width=\textwidth]{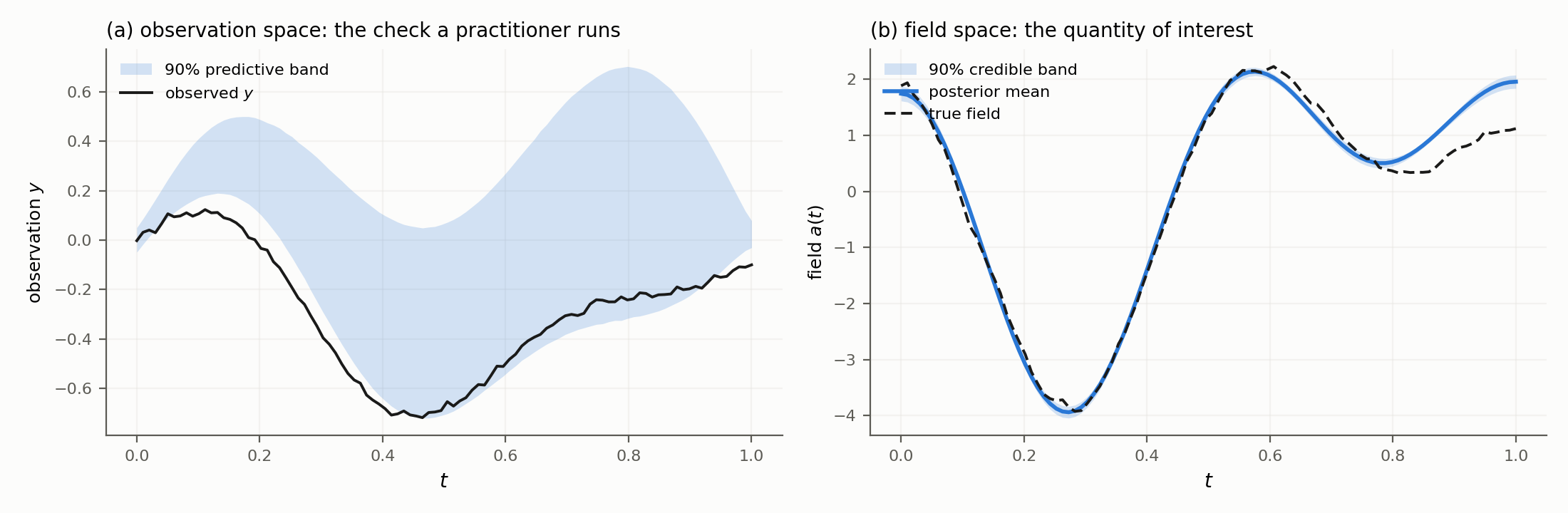}
\caption{The dissociation on a single fit: one model, one dataset, one nominal
level. (a) In observation space the predictive band covers the data and the check
passes. (b) The field posterior of the same model for the same case. The band is
narrow and the truth lies outside it over much of the domain. A practitioner sees
(a); the quantity of interest is (b). Each panel shows the test case at the
median field RMSE among the 100 held-out cases for its configuration.}
\label{fig:twospaces}
\end{figure}

The ratio $\rho=\sigma_\perp^2/\sigma^2$, where $\sigma_\perp^2$ is the
per-dimension observation energy outside $\range(H\Phi)$, would reveal the
misspecification. Computing it requires the true noise scale. The
fitted scale cannot stand in for it, because it has absorbed the excess.
Absorption matches the magnitude of the complement residual but not its shape:
structured out-of-basis energy makes the complement covariance anisotropic.

The method needs repeated realisations: $N$ observation vectors through one
operator $H$ and one basis $\Phi$, with an independent field behind each. The
basis must be fixed independently of the realisations being tested. The test
does not apply to a single observation vector.

Our contributions are the following.
\begin{itemize}\itemsep2pt
\item We identify absorption of the noise scale as the channel through which
basis misspecification escapes observation-space checks
(Proposition~\ref{prop:absorption}). Pinning $\sigma$ at its true value removes
the concealment but not the collapse of field coverage.
\item We show that, conditionally on the fitted noise scale, every exact test on
the complement is a test of the scale-free direction of the residuals
(Proposition~\ref{prop:forced}). We test the shape of their sample spectrum with
John's statistic. Its null model and null mean are exact at finite sample size
(Proposition~\ref{prop:sphericity}, Lemma~\ref{lem:nullmean}), and its power at
proportional dimension is derived for Gaussian data (Theorem~\ref{thm:prop}).
Its blind spot is Gaussian energy that is isotropic in the complement.
\item We measure power over $1047$ simulated configurations and in a
preregistered study of GEBCO topography under three forward operators.
\end{itemize}

John's statistic, its high-dimensional theory and its use in invariant
detection are classical (Section~\ref{sec:related}). What is new is the absorption
mechanism in basis-restricted inverse problems, the reason the residual direction
is the information that remains, and the link from that direction to credible-set
failure in field space. Table~\ref{tab:status} gives the status of each result.

\section{Setup}\label{sec:setting}

Let the physical domain be discretised on $M$ points, with $a\in\R^M$ the field
and $y\in\R^n$ the observations. The forward operator is linear, and
\begin{equation}
  y = Ha + \varepsilon, \qquad \varepsilon\sim\Nor(0,\sigma^2 I_n).
\end{equation}
We observe $N$ vectors $y_1,\dots,y_N$ from one forward problem, with one
operator, one noise scale and one basis, and an independent field realisation
behind each. Examples are tomographic scans of many specimens through one scanner, strain
fields on one test rig, and an airborne gravimeter flown over many survey tiles
(Section~\ref{sec:real}).

Write $\Phi\in\R^{M\times d}$ for the basis. We assume that $\Phi$ is fixed
independently of the $N$ realisations tested. A basis learned from the same
realisations makes the complement data-dependent, and the finite-sample null is
then no longer exact.

The true field decomposes as $a = \Phi w^\star + \xi_\perp$ with
$\xi_\perp\perp\spn(\Phi)$, and the model is well specified if and only if
$\xi_\perp=0$. The fitted prior writes $a=\Phi w$. The theory needs only a coefficient prior
that is Gaussian given a latent label; the Gaussian-mixture prior of our
experiments has this form (proof of Proposition~\ref{prop:absorption}). The fitted
model also has a basis-restricted Gaussian process residual with covariance
$\Phi S'\Phi^\top$, so its image under $H$ lies in $\range(G)$
(Appendix~\ref{app:setup}). The coefficient prior matters little: the mixture adds
little to reconstruction beyond the restriction itself, with a measurable gain on
one of nine benchmarks (Appendix~\ref{app:negative}). The contribution concerns the
basis restriction, the fitted noise scale and the complement test.

In observation space the model lives in $G := H\Phi$, with $r := \rank(G)$. Let
$\Ppar$ project onto $\range(G)$ and $\Pperp=I-\Ppar$. Given the latent label, the
model's marginal covariance of $y$ is $K_y=GSG^\top+\sigma^2I_n$, where $S$ is the
coefficient covariance plus $S'$. Define
\begin{equation}
  \sigma_\perp^2 := \tfrac{1}{n-r}\tr\!\left[\Pperp\left(HK_\perp H^\top+\sigma^2 I\right)\Pperp\right],
\end{equation}
with $K_\perp=\Cov(\xi_\perp)$. This is the per-dimension observation energy the
model cannot see; its counterpart on $\range(G)$ is
$\sigma_\parallel^2:=\tfrac1r\tr[\Ppar(HK_\perp H^\top+\sigma^2I)\Ppar]$. We
assume $\Pperp H\,\E[\xi_\perp]=0$, which holds when $\xi_\perp$ has mean zero. The
misspecification ratio is $\rho := \sigma_\perp^2/\sigma^2$. It is estimable
given $H$, $\Phi$ and a known $\sigma$, with no ground-truth field. The observable law depends on the coefficient covariance and the
basis-restricted residual only through their sum, so we report no separate
covariance recovery.

\section{Absorption of the Noise Scale}\label{sec:theory}

\begin{proposition}[Absorption]\label{prop:absorption}
Under the model of Section~\ref{sec:setting} with linear $H$ and
basis-restricted covariance, $\Pperp G=0$ and $\Pperp K_y \Pperp =
\sigma^2\Pperp$, so the log-likelihood separates as $\ell(\Theta) =
\ell_\parallel(\Theta) + \ell_\perp(\sigma)$ with
\begin{equation}
  \ell_\perp(\sigma) = -\tfrac{n-r}{2}\log(2\pi\sigma^2) - \tfrac{\|\Pperp y\|^2}{2\sigma^2}.
\end{equation}
If the in-span prior scale dominates the noise, so that the dependence of
$\ell_\parallel$ on $\sigma$ is negligible, the maximum-likelihood estimate is
$\hat\sigma^2 = \|\Pperp y\|^2/(n-r)$, which has expectation $\sigma_\perp^2$, so
$\hat\sigma^2/\sigma^2=\rho$ in expectation and $\hat\sigma/\sigma\approx\sqrt{\rho}$,
with no free parameters.
\end{proposition}

\begin{proof}
Let $U_\parallel$ and $U_\perp$ be orthonormal bases of $\range(G)$ and its
complement, so that $\Ppar=U_\parallel U_\parallel^\top$ and
$\Pperp=U_\perp U_\perp^\top$. We distinguish the model law, used in the
likelihood, from the true law of $y$, used in the expectation.

\emph{Model covariance.} Conditionally on the latent label (a mixture component, or
more generally the variable on which the coefficient prior is Gaussian), the model
has $y\sim\Nor(Gm,\,K_y)$ with $K_y=GSG^\top+\sigma^2I_n$ (Section~\ref{sec:setting}),
where $S$ collects the coefficient covariance and the basis-restricted residual,
whose image under $H$ lies in $\range(G)$. Since $\Pperp G=0$,
$\Pperp K_y\Ppar=\sigma^2\Pperp\Ppar=0$ and $\Pperp K_y\Pperp=\sigma^2\Pperp$. In the coordinates $[U_\parallel,U_\perp]$ the
covariance is therefore block-diagonal, and $U_\perp^\top Gm=0$, so $U_\perp^\top y$
has mean zero under the model.

\emph{Separation.} The change of variables to $[U_\parallel,U_\perp]$ is orthogonal,
so each conditional density factorises as
$\Nor(U_\parallel^\top y;U_\parallel^\top Gm,U_\parallel^\top K_yU_\parallel)\,
\Nor(U_\perp^\top y;0,\sigma^2I_{n-r})$. The second factor is the same for every
component of the Gaussian mixture, so it factors out of the mixture sum; the same
holds for an integral over any conditionally Gaussian prior. Taking logarithms
gives $\ell=\ell_\parallel(\Theta)+\ell_\perp(\sigma)$, with $\ell_\perp$ as stated,
because $\|U_\perp^\top y\|=\|\Pperp y\|$. The separation is exact.

\emph{Maximiser.} $\ell_\parallel$ depends on $\sigma$ only through
$U_\parallel^\top K_yU_\parallel=U_\parallel^\top GSG^\top U_\parallel+\sigma^2I_r$.
This step needs $G$ to have full column rank. Then, if the in-span covariances are
unconstrained, the range covariance $V_k$ of each mixture component ranges over
every $V\succeq\sigma^2I_r$, so the supremum of $\ell_\parallel$ over the in-span
parameters does not depend on $\sigma$ while
$\sigma^2\le\lambda^\ast:=\min_k\lambda_{\min}(\hat V_k)$, where $\hat V_k$ are the
unconstrained maximising range covariances of the components. For a mixture this
is a condition on each component, not on the marginal $\Cov(U_\parallel^\top y)$. When
$\|\Pperp y\|^2/(n-r)\le\lambda^\ast$, which is what dominance of the in-span prior
scale means here, setting $\partial\ell_\perp/\partial\sigma=0$ gives the exact
maximiser $\hat\sigma^2=\|\Pperp y\|^2/(n-r)$; otherwise $\ell_\parallel$ also moves
$\hat\sigma$ and the identification is approximate. With $N$ realisations the
log-likelihoods add, so $\hat\sigma^2$ is the average of $\|\Pperp y_i\|^2/(n-r)$.

\emph{Expectation.} Under the true law, $y=H\Phi w^\star+H\xi_\perp+\varepsilon$ with
$\varepsilon\sim\Nor(0,\sigma^2I_n)$ independent of $(w^\star,\xi_\perp)$. The
identity $\Pperp H\Phi=\Pperp G=0$ removes every term involving $\Phi$, so
$\Pperp y=\Pperp H\xi_\perp+\Pperp\varepsilon$ and
\[
  \E\|\Pperp y\|^2
  =\tr\!\big[\Pperp(HK_\perp H^\top+\sigma^2I)\Pperp\big]
  +\big\|\Pperp H\,\E[\xi_\perp]\big\|^2 .
\]
Under the assumption $\Pperp H\,\E[\xi_\perp]=0$ of Section~\ref{sec:setting},
$\|\Pperp y\|^2/(n-r)$ has expectation $\sigma_\perp^2$ by the definition of
$\sigma_\perp^2$. The
identification of the maximiser needs the dominance condition.
\end{proof}

Across $31$ sweep points the measured $\hat\sigma/\sigma$ (of the implemented noise
update, Appendix~\ref{app:absorption}) matched $\sqrt\rho$ to within $3.4\%$ with no trend in the residual, and to $0.5\%$
after a finite-$n$ degrees-of-freedom correction (Appendix~\ref{app:absorption}).
Proposition~\ref{prop:absorption} characterises the exact maximiser of the stated
likelihood. The implemented update maximises $\log p(y\mid\hat w)$ at the
posterior-mean coefficient and normalises by $n$ rather than $n-r$. Its $\hat\sigma$ is
therefore approximately $\sqrt{(n-r)/n}$ times that of the exact maximiser, a factor
that tends to $1$ as $r/n\to0$ (Appendix~\ref{app:absorption}).

\section{A Scale-Free Test on the Complement}\label{sec:sphericity}

\begin{table}[!t]
\caption{Status of the main-text results. \textsc{exact} holds at every finite
$(N,p)$ with no asymptotics and no unstated condition; \textsc{approximation} is
exact in a named regime and approximate outside it, with the measured departure
given; \textsc{empirical} was measured over the stated grid and is not
extrapolated beyond it. Table~\ref{tab:status-full} (Appendix~\ref{app:status}) covers every result
in the paper.}
\label{tab:status}
\begin{center}\small
\begin{tabular}{>{\raggedright\arraybackslash}p{0.30\textwidth}%
                >{\raggedright\arraybackslash}p{0.16\textwidth}%
                >{\raggedright\arraybackslash}p{0.46\textwidth}}
\toprule
result & status & condition it rests on \\
\midrule
Prop.~\ref{prop:sphericity}(a)--(d): scale invariance, the null model, detectability, the blind spot & \textsc{exact} & Gaussian $\varepsilon$. (b) needs $\xi_\perp=0$ and the model mean in $\range(G)$; (c) needs $A$ with $k$ equal nonzero eigenvalues; (d) needs Gaussian $\xi_\perp$ and $A=(\rho-1)I_p$. (c) and (d) are statements about $\operatorname{median}(\Rstd\mid H_0)$, not its mean \\
Prop.~\ref{prop:forced}: given $\|Z\|_F^2$, every exact test on the complement is a test of direction & \textsc{exact} & Gaussian $\varepsilon$; completeness of $\|Z\|_F^2$ and Neyman structure \citep{lehmann2005testing} \\
Calibration of that null by parametric bootstrap & \textsc{empirical} & Monte-Carlo, not exact; $1000$ replicates at each $(N,p)$; pooled false-positive rate $0.0509$ \\
Lem.~\ref{lem:nullmean}: $\E[U\mid H_0]=(p-1)(p+2)/(\nu p+2)$ & \textsc{exact} & Gaussian $\varepsilon$; every $(\nu,p)$, $p>\nu$ included; median $0.041\%$ over $59$ pairs \\
Prop.~\ref{prop:absorption}: $\hat\sigma/\sigma=\sqrt\rho$ & \textsc{approximation} & separation is exact; the maximiser needs the in-span prior scale to dominate the noise. The implemented update normalises by $n$ (App.~\ref{app:absorption}); matched to $3.4\%$, or $0.5\%$ with the $\sqrt{(n-r)/n}$ factor \\
False-positive rate $0.0509$; clause (d) at the nominal rate & \textsc{empirical} & $98$ null cells at $59$ $(N,p)$ pairs, $19\,600$ draws; $98$ isotropic cells, detection $0.045$ \\
\bottomrule
\end{tabular}
\end{center}
\end{table}

Let $U_\perp$ be an orthonormal basis for the complement of $\range(G)$, let
$p:=n-r$, and let $z:=U_\perp^\top y$. Then $\Cov(z)=\sigma^2(I_p+A)$ with
$A:=U_\perp^\top HK_\perp H^\top U_\perp/\sigma^2$ and $\tr(A)/p=\rho-1$. From
$N$ held-out observations let $\lambda_1\ge\dots\ge\lambda_p$ be the sample
covariance eigenvalues and $\bar\lambda$ their mean, and define
\begin{equation}\label{eq:stats}
  U := \tfrac1p\textstyle\sum_{j}\big(\tfrac{\lambda_j}{\bar\lambda}-1\big)^{2},
  \quad
  \hat\rho := \tfrac{\bar\lambda}{\operatorname{median}(\lambda)},
  \quad
  R := \tfrac{\lambda_1}{\operatorname{median}(\lambda)} .
\end{equation}
The test statistic is $U$, John's sphericity statistic \citep{john1971some}.
$\hat\rho$ is an auxiliary estimate of the size of the excess, and $R$ is the
ratio of the leading eigenvalue to the bulk. We standardise $R$ by its null
median,
\begin{equation}\label{eq:Rstd}
  \Rstd := R/R_0(N,p),
  \quad
  R_0(N,p) := \operatorname{median}\big(R \mid \Cov(z)\propto I_p\big),
\end{equation}
with $R_0$ from the same bootstrap that calibrates $U$, so no $\sigma$ is needed.
Appendix~\ref{app:aux} discusses $\hat\rho$, $R$ and $\Rstd$ and gives
detection thresholds in $\Rstd$.

Write $Z\in\R^{N\times p}$ for the stacked complement residuals.

\begin{proposition}[Exact tests are conditionally tests of direction]\label{prop:forced}
Under correct specification the rows of $Z$ are i.i.d.\ $\Nor(0,s^2I_p)$ with
$s=\sigma$ unknown, and $T:=\|Z\|_F^2$ is complete sufficient for $s$.
\textbf{(i)} Every test with level exactly $\alpha$ at every $s>0$ satisfies
$\E[\phi\mid T]=\alpha$, so for almost every value of $T$ it is a level-$\alpha$
test of the direction $Z/\|Z\|_F$, whose null law is free of $s$.
\textbf{(ii)} A magnitude referred to the noise scale fitted on the same
complement is constant: $T/(Np\hat s^2)\equiv1$ for $\hat s^2=T/(Np)$.
\textbf{(iii)} If the out-of-span content leaves the complement law
$\Nor(0,s'^2I_p)$ for some $s'$, no statistic of $Z$ has power above $\alpha$.
\end{proposition}

\begin{proof}
(i) $T$ is complete in a one-parameter exponential family, so an exact test has
Neyman structure \citep[Theorem~4.3.2]{lehmann2005testing}: $\E[\phi\mid T]=\alpha$.
$Z/\|Z\|_F$ is independent of $T$ (Appendix~\ref{app:nullmean}), so conditionally
on $T$ the test is a test of the direction, whose null law is free of $s$. An
exact test may still depend on $T$, but only through which test of direction it
runs at each value of $T$. (ii) follows from the definition of $\hat s^2$.
(iii) Such an alternative lies in the null family.
\end{proof}

Reducing the direction further to the shape of the sample
spectrum uses invariance under rotation and scaling: with multiple snapshots the
maximal invariant is the set of sample eigenvalues divided by their trace
\citep{besson2006}.

Only information about $\sigma$ that absorption cannot reach, such as a known
instrument noise level, repeat measurements or noise-only runs, escapes
Proposition~\ref{prop:forced}. Our setting has none. For a single Gaussian with unconstrained coefficient covariance, $\range(G)$ adds only the one-sided bound
$\sigma^2\le\lambda_{\min}\big(\Cov(U_\parallel^\top y)\big)$, which binds only
where the dominance condition of Proposition~\ref{prop:absorption} fails.

\begin{proposition}[Scale-free detection of orthogonal unmodelled energy]\label{prop:sphericity}
\textbf{(a) Scale invariance.} $U$, $\hat\rho$, $R$ and $\Rstd$ are invariant
under $z\mapsto cz$ for every $c>0$, so all four are computable without knowing
or estimating $\sigma$. Any statistic of the form $\sigma_\perp^2/\hat\sigma^2$
degenerates to $1$ under Proposition~\ref{prop:absorption}; a scale-free
functional of the spectrum cannot.

\textbf{(b) Exact null model.} If $\xi_\perp=0$ and the model mean lies in
$\range(G)$ then $z=U_\perp^\top\varepsilon$ exactly, so $A=0$ and
$\Cov(z)=\sigma^2I_p$. Under Gaussian $\varepsilon$ this holds at every finite
$(N,p)$, with no asymptotic approximation and no condition on $p/N$. The
\emph{null model} is therefore exact. The test as implemented calibrates it by
parametric Monte-Carlo bootstrap at the matching $(N,p)$, which carries the usual
Monte-Carlo error and is not exact. One bootstrap delivers the null law of $U$
and the level $R_0(N,p)$ together.

\textbf{(c) Detectability.} If $A$ has $k$ equal nonzero eigenvalues then
$\tr(A)=p(\rho-1)$ forces each to $p(\rho-1)/k$, so the leading eigenvalues of
$\Cov(z)$ sit at $\sigma^2(1+(\rho-1)p/k)$: detectability is governed by
$(\rho-1)p/k$, not by $\rho$ alone. That quantity requires $\sigma$ and a $k$.
Its observable counterpart $\Rstd$ requires neither, and
$\operatorname{median}(\Rstd\mid H_0)=1$ at every finite $(N,p)$ by
construction. The statement for the mean does not follow and is not used: the
test runs on $U$, and $\Rstd$ only reports the detection threshold.

\textbf{(d) Blind spot.} If $\xi_\perp$ is Gaussian and $A=(\rho-1)I_p$, so the
unmodelled energy is white in the complement, then $\Cov(z)=\sigma^2\rho I_p$ is
spherical for every $\rho$. Then $U$, $\hat\rho$ and $\Rstd$ have exactly their
null distributions, so the test has no power at any $\rho$. In particular
$\operatorname{median}(\Rstd\mid A=(\rho-1)I_p)=1$.
\end{proposition}

\begin{proof}
(a) Each statistic is a ratio of spectral functionals of equal degree, so the
factor $c^2$ cancels; $\sigma_\perp^2/\hat\sigma^2$ degenerates by
Proposition~\ref{prop:absorption}. (b) With $\xi_\perp=0$ and the mean in
$\range(G)$, $z=U_\perp^\top\varepsilon$ and $U_\perp$ has orthonormal columns.
(c) $\tr(A)=p(\rho-1)$ shared equally over $k$ eigenvalues gives each
$p(\rho-1)/k$; $\operatorname{median}(\Rstd\mid H_0)=1$ because $R_0$ is defined as
the null median of $R$. (d) Clause (d) is the null family of
Proposition~\ref{prop:forced}(iii): $z$ is Gaussian with a spherical covariance,
so the scale-invariant statistics have their null laws.
\end{proof}

Non-Gaussian energy that is isotropic in covariance lies outside the null family
(Section~\ref{sec:real}). Clause (c) concerns $k$ equal spikes: only then does
$(\rho-1)p/k$ govern detection. Clause (b) makes the null model exact; the test
calibrates that null by bootstrap, which carries Monte-Carlo error.

\begin{lemma}[Exact null mean]\label{lem:nullmean}
Under $H_0$ of Proposition~\ref{prop:sphericity}(b), with $\nu=N-1$ degrees of
freedom,
\begin{equation}\label{eq:nullmean}
  \E[\,U \mid H_0\,] \;=\; \frac{(p-1)(p+2)}{\nu p+2}.
\end{equation}
\end{lemma}

The proof is in Appendix~\ref{app:nullmean}. Across $59$ distinct $(N,p)$ pairs
the formula matched the bootstrapped mean to a median relative error of
$0.041\%$. The null law of $U$ moves with $(N,p)$, so the bootstrap is
recalibrated at each pair.

\paragraph{Decision rule.} The test compares $U$ with its bootstrap null at the
matching $(N,p)$ and works at any $p/N$. At $N=5$ and $p=96$, where the sample
covariance has rank $4$, detection was $0.810$ with the false-positive rate at
nominal. A rejection means a detectable departure of the complement from white
Gaussian noise, caused by omitted field energy or by correlated,
heteroskedastic or non-Gaussian noise. It does not by itself show harmful
reconstruction error: on real topography all four cells without harm also fire.
Failure to reject means that no departure was detected; the test is silent about
white Gaussian complement energy and about energy mapped into $\range(G)$.

\section{Synthetic Experiments}\label{sec:experiments}

\begin{figure}[!t]
\centering
\includegraphics[width=\textwidth]{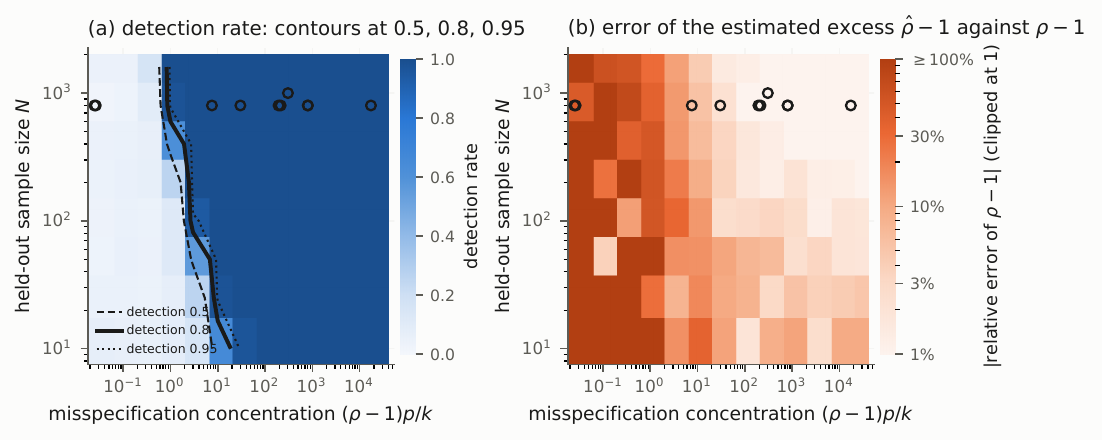}
\caption{Power surface for a single spike at $p=96$, $500$ replicates per cell.
(a) Detection rate, with contours at $0.5$, $0.8$ and $0.95$. (b) Relative error
of the estimated excess $\hat\rho-1$ against $\rho-1$. Simulated cells sit at the population excess $(\rho-1)p/k$. Circles mark the nine
benchmarks and the sweep operating point, placed at their measured excess above the
null, $R-R_0(N,p)$, as a proxy; at the sweep point the two agree ($30.09$ against
$30.72$). The two left of the $0.8$ contour are the correctly specified
benchmarks.}
\label{fig:power}
\end{figure}

Nine benchmarks span integral operators, deconvolution, sensor subsampling and a
2D elliptic PDE, with $p$ from $16$ to $135$. A sweep of $31$ points varies
$\rho$ over a factor of $306$ on one operator family (Appendices~\ref{app:setup}
and~\ref{app:sweep}).

\paragraph{Dissociation and benchmarks.} Over the sweep, field coverage fell from
$0.851$ to $0.331$ while observation coverage stayed near nominal. Pinning $\sigma$ at its true value removed the concealment, not the collapse
(Appendices~\ref{app:coverage} and~\ref{app:sweep}). The test is calibrated
against an exact null. At $\rho=1.32$ it detected at $1.00$, where field
coverage had already fallen to $0.547$. The misspecified benchmarks sit one to four orders of magnitude above the
detection band ($1.17$--$1.52$; Appendix~\ref{app:aux}) and do not stress the test
(Appendix~\ref{app:practice}).

\paragraph{Power.} We drew the complement residual from the law of
Proposition~\ref{prop:sphericity}, specifying only the eigenvalues of $A$; for a
statistic of the sample spectrum this covers every Gaussian alternative. The grid
has $k$ equal spikes for $k$ from $1$ to $p$, polynomial and exponential decay,
and isotropy, with $p$ from $16$ to $250$ and $N$ from $3$ to $3200$: $1047$
cells (Figure~\ref{fig:power}). The false-positive rate over the null cells was $0.0509$ against a
nominal $0.05$ (Appendix~\ref{app:power}). Across $p/N$ from $0.01$ to $10$,
$39$ of $40$ misspecified cells detected at $1.000$, including all six with
$p>N$. The null is calibrated by bootstrap at each $(N,p)$, so a singular sample
covariance does not invalidate the test; the limiting null law of
\citet{liyao2016} supports this asymptotically.

\paragraph{Blind spot.} Over $98$ isotropic cells the detection rate had median
$0.045$, the nominal rate, as clause (d) predicts. At $(N,p)=(800,96)$ and $\rho=1.32$ the test kept power of $0.8$
until the energy was spread over more than $83\%$ of the complement directions.
At $N=100$ that width fell to $50\%$.

\paragraph{Power law.} Theorem~\ref{thm:prop} (Appendix~\ref{app:propthm}) gives
the law of $U$ at proportional $p/\nu$ for Gaussian data and any spectrum of
bounded norm. For $\nu\ge10$ it predicts detection to a mean absolute error of $0.011$, and
$0.008$ with a third-order term.

\section{Real Topography}\label{sec:real}

The simulations fix the spectra. This section measures whether real
misspecification departs from isotropy.

\paragraph{Data and design.} The field is GEBCO 2024 topography
\citep{gebco2024} at $15$ arc-seconds, restricted to directly measured cells and cut into $32\times32$
tiles ($M=1024$). Whole $5^\circ$ blocks go to a basis split or an evaluation
split, with a $32$-cell buffer discarded at every block boundary. The basis is
built from $65$ blocks and $22\,980$ tiles and tested on $23$ blocks and
$8\,108$ tiles. $\Phi$ is the leading-$d$ principal basis of the basis split,
for $d$ in $\{4,8,16,32,64,96\}$. The operators are linearised Bouguer gravity at
$14\times14=196$ stations, its vertical gradient, and pointwise subsampling. With
two noise scales this gives $36$ cells. The design and its admissible outcomes
were preregistered before any data were read; the evaluation set was enlarged after
registration (Appendix~\ref{app:real}).

\paragraph{The complement.} The participation ratio
$\mathrm{PR}=(\sum_j a_j)^2/\sum_j a_j^2$ of the population complement excess
equals $p$ under isotropy and otherwise counts the directions that carry the
excess. Over the $18$ operator-and-$d$ cells
$\mathrm{PR}/p$ runs from $0.036$ to $0.709$, so on this field the
misspecification concentrates. $\mathrm{PR}/p$ rises with $d$ under every
operator. At matched $d$ it is smallest for gravity and largest for pointwise
subsampling, the order in which the operators stop suppressing high wavenumbers.

Per-tile complement energy is also heavy tailed, with a coefficient of variation
of median $2.85$ against $0.108$ for Gaussian noise. Randomising each tile's
complement direction while keeping its norm makes the population covariance
exactly isotropic, which clause~(d) would make undetectable if it were Gaussian.
The test still detects at $0.95$ or more at $N=100$ in $35$ of $36$ cells. A
Gaussian with the same covariance keeps only the anisotropy, and the test detects
at $0.95$ or more in $34$ of $36$. The test does not say which departure it saw.

\paragraph{Detection, harm and competitors.} Absorption reproduces. On one half of every block, $\hat\sigma$ is the closed-form
complement estimate, which equals the maximum-likelihood value because the
range-block bound never binds (Appendix~\ref{app:real}); against $\rho$ from the
other half it gives a median ratio of $1.009$ over the $36$ cells, range $0.998$ to
$1.016$. The
agreement weakens with geographic separation, tracking the $15.7\%$ by which the
excess is larger where the basis was not trained (Appendix~\ref{app:real}).
Observation coverage stays in $0.924$--$0.978$ over all $36$ cells while field
coverage runs from $0.583$ to $0.954$, and $\hat\sigma/\sigma$ reaches $185$.

A cell is harmed if field coverage falls below $0.85$ or field RMSE exceeds
$1.2$ times the in-span floor. Of the $32$ harmed cells, $31$ are detected at
$N=100$ and all $32$ at $N\ge800$. The single miss is gravity at $d=96$, where
the population excess is $4\%$. The four unharmed cells also fire; all four have $\rho>1$, so these are not false
positives, and specificity is not estimable on these data. $\Upop$, the statistic $U$ evaluated on the population spectrum, is an exact
function of $\rho$ and $\mathrm{PR}$
(Proposition~\ref{prop:fielderror}). Computed from the basis split with the
nominal $\sigma$, it gives the right fire/no-fire call in $35$ of $36$ cells.

Cross-validated predictive coverage and held-out observation coverage fire on
$8$ and $5$ of the $32$ harmed cells. A per-index calibration check fires
everywhere until it is given a null, and then fires nowhere. These three read
only the diagonal of the complement covariance. John's dispersion $\Upop$ computed
from that diagonal instead of the eigenvalues is, over the $36$ cells, a
median $128$ times smaller than on the eigenvalues. A
full-covariance goodness-of-fit test with $\Sigma_y$ fixed from training fires
in all $36$ cells, including the unharmed ones, but it is not scale-free.
Inflating the field variance by $15\%$, with the shape exactly correct, makes it
reject on $99\%$ of draws (Appendix~\ref{app:real}).

Although neighbouring tiles are dependent, the test stays calibrated on them. Its
rejection rate under an exact null built from the real tiles, with the out-of-span
content removed so that the complement is $U_\perp^\top\varepsilon$, gives
a false-positive rate of median $0.040$ against a nominal $0.05$ over $18$
distinct measurements (six thinning levels, three complement dimensions). Leakage into $\range(G)$ grows with $d$
(Appendix~\ref{app:fielderror}), and the range-block bound never binds
(Appendix~\ref{app:real}).

\section{Related Work}\label{sec:related}

\paragraph{Credible sets in inverse problems.} \citet{knapik2011bayesian} showed
for linear inverse problems with Gaussian priors that coverage depends jointly on
prior and truth regularity, with asymmetric failure modes.
\citet{szabo2015frequentist} studied adaptive constructions. Reduced-rank and
constrained Gaussian process priors
\citep{solin2020hilbert,jidling2017linearly,spantini2015optimal} supply the
restriction we study. Predictive checks in observation space
\citep{gelman1996posterior} are the diagnostics the mechanism defeats.

\paragraph{Variance absorption under misspecification.}
\citet{bachoc2013cross} analysed maximum-likelihood estimation of Gaussian
process hyperparameters under a misspecified covariance, and
\citet{bachoc2018asymptotic} gave an asymptotic treatment. That a variance
parameter absorbs structure the covariance cannot represent is therefore known
for Gaussian process regression. Here the absorbing parameter lives in
observation space and the failure it conceals lives in field space; in direct
regression the two coincide. Pinning $\sigma$ shows that absorption is the
channel between the two. We also give a rate where the existing results are
qualitative or asymptotic. The likelihood separation is exact. The
identification $\hat\sigma/\sigma\approx\sqrt\rho$ is an approximation that
needs the in-span prior scale to dominate the noise (Table~\ref{tab:status}).
With a degrees-of-freedom correction it matched to $0.5\%$ across two and a half
decades, with no fitted constant.

\paragraph{Lack-of-fit and residual diagnostics.} Classical lack-of-fit tests
compare residual variation with pure error from repeated measurements of the
same unit, or from clusters of near replicates \citep{christensen1989}. Either
source supplies a $\sigma$ the fit cannot touch; our realisations, different
fields through one design, supply none. Residual-pattern checks for inverse
problems also need $\sigma$, including a recent one motivated by the same
attenuation of structured model error \citep{kazlauskaite2026}. These checks
need an external $\sigma$ or read residual magnitude, and absorption defeats
both here: referred to a same-sample $\hat\sigma$, any magnitude check on the
complement is constant (Proposition~\ref{prop:forced}(ii)).

\paragraph{Sphericity and invariant detection.} $U$ is John's statistic
\citep{john1971some}, with null and spiked-alternative theory from
\citet{ledoitwolf2002}, \citet{wangyao2013}, \citet{liyao2016} and
\citet{onatski2013}. Invariant detectors project out a known interference subspace and test
scale-invariantly \citep{scharf1994}; with multiple snapshots the maximal
invariant is our spectrum shape \citep{besson2006}.

\paragraph{Error contrasts and prior--data conflict.} The complement residual is
a vector of error contrasts \citep{theil1965,patterson1971}. Checking the
sampling model through the law of the data given a minimal sufficient statistic
is the principle of \citet{evans2006}.

\paragraph{What this paper adds.} Proposition~\ref{prop:forced} reaches the
direction $Z/\|Z\|_F$ through completeness rather than invariance; invariance
then reduces it to the spectrum shape. The absorption law explains why the
diagnostics practitioners use fail. The interference subspace is the model's own range,
$\range(H\Phi)$, fixed by the forward operator and the basis. Theorem~\ref{thm:prop} gives a power law for general spectra at
proportional $p/\nu$ for Gaussian data, which neither \citet{wangyao2013} nor
\citet{liyao2016} covers. We did not find the construction applied to checking a
basis-restricted prior.

\section{Discussion and Limitations}\label{sec:discussion}

On our nine benchmarks, the rank-$d$ restriction beat a full-rank prior on the
two with $\rho\approx1$, by $8.2\times$ and $3.9\times$ in field RMSE, and lost on the seven
out-of-basis problems, as the oversmoothing dichotomy predicts. We propose no
threshold on $\rho$ (Appendix~\ref{app:scoping}).

The test cannot see error inside the basis. On the Bimodal benchmark our model
and Gaussian process regression with the same kernel both sit well above the
in-span floor $\|a-P_{\spn(\Phi)}a\|_{\mathrm{rms}}$, the best error any
basis-restricted estimator could achieve. That excess lies inside $\spn(\Phi)$,
where no statistic on the complement can reach it: reducing complement anisotropy
does not remove in-span estimation error (Appendix~\ref{app:negative}).

The real-data evidence is one field, one basis construction and three operators.
A field whose out-of-basis content is Gaussian and unstructured would land in
the blind spot of clause~(d). The real-data study does not rule out this case; the one registered field we
examined was not such a field: its complement excess is anisotropic and its
per-tile complement energy is heavy tailed (Section~\ref{sec:real}). Specificity is not estimable on those data, because every
cell is misspecified to some degree, so the false-positive side of the claim
rests on the exact null. The power study measures power against the Gaussian family of
Proposition~\ref{prop:sphericity}, not the law of real misspecification. Non-Gaussian
out-of-span signal is covered by the real-data arms. Non-Gaussian noise
$\varepsilon$ is not, since it would change the null as well as the alternative.

Everything rests on $\Pperp H\Phi=0$, which requires a linear and known $H$.
Extension to nonlinear operators is open. The power law is proved for Gaussian
data only; non-Gaussian behaviour is measured, not derived.

\subsubsection*{AI Use Statement}

Generative AI tools were used during preparation of this manuscript for conceptual discussion and development, mathematical proof checking, literature-search assistance, drafting, editing, and organization. The authors reviewed and verified all AI-assisted material used in the manuscript, including mathematical arguments, citations, and interpretations. Generative AI was not used to generate experimental data or to determine the reported experimental results. The authors take full responsibility for the final content of the submission, including all mathematical statements, experimental results, citations, and interpretations. Generative AI systems are not authors of this work.

\subsubsection*{Reproducibility Statement}

All claims are supported by committed artifacts. The data-generating process,
model, training procedure and evaluation protocol are specified in
Appendix~\ref{app:setup}; all training uses a convergence criterion rather than a
fixed budget, and every reported metric uses a held-out split. Every figure and
table is produced by a released script from a stored JSON artifact, and a
consistency checker pins each quantitative claim in this paper to the artifact it
came from and fails on any disagreement. The checker reports a missing artifact
as pending rather than passing it over, so a partial checkout is visibly
incomplete rather than quietly green. Seeds are fixed and recorded per run.
Code and artifacts are provided in the supplementary material; the data are
regenerated by two scripts in the release rather than shipped.

\bibliography{references}
\bibliographystyle{tmlr}

\appendix
\section*{Appendix}

The appendices follow the order of the main text. Each opens with the claim it
supports.

\section{The Absorption Law in Detail}\label{app:absorption}

This appendix supports Proposition~\ref{prop:absorption} and its validation in
Section~\ref{sec:theory}. The law holds across the sweep, and its small residual
is a finite-$n$ degrees-of-freedom effect.

\begin{table}[htbp]
\caption{Absorption across the sweep: measured $\hat\sigma/\sigma$ against
$\sqrt\rho$, with no fitted constant.}
\label{tab:absorption-app}
\begin{center}\small
\begin{tabular}{rrrr}
\toprule
$\rho$ & $\sqrt{\rho}$ & $\hat\sigma/\sigma$ & rel.\ error \\
\midrule
1.00 & 1.00 & 0.98 & 1.7\% \\
4.36 & 2.09 & 2.04 & 2.0\% \\
8.33 & 2.89 & 2.81 & 2.6\% \\
15.48 & 3.93 & 3.82 & 3.0\% \\
27.97 & 5.29 & 5.13 & 3.1\% \\
49.05 & 7.00 & 6.78 & 3.2\% \\
116.00 & 10.77 & 10.42 & 3.3\% \\
306.01 & 17.49 & 16.99 & 2.9\% \\
\bottomrule
\end{tabular}

\end{center}
\end{table}

\begin{remark}[Where the in-span variance goes]\label{rem:breakdown}
The measured $\hat\sigma/\sigma$ matches
$\sqrt{\rho}$ to within $3.4\%$ at every one of the 31 points, with no trend in
the residual (Table~\ref{tab:absorption-app}). All runs use a convergence
criterion on the marginal likelihood. The EM optimiser needs $470$ epochs to
converge at $\rho=306$ against $63$ at $\rho=1$, and a fixed epoch budget tuned on
the well-specified case stops $\hat\sigma$ $65\%$ short of its maximiser at
$\rho=306$. Tightening the convergence criterion by a further factor of ten moves
the residual only from $3.0\%$ to $2.6\%$.

The residual is a finite-$n$ effect. Holding $\rho$ fixed and varying
$n\in\{25,50,100,200\}$, the measured shortfall is
$-7.7\%,-3.7\%,-2.4\%,-1.3\%$, tracking $-r/(2n)$ at every $n$, which is
$\sqrt{(n-r)/n}-1$ to first order. Section~\ref{sec:setting} normalises
$\sigma_\perp^2$ by the complement dimension $n-r$. The implemented noise update is
not the exact maximiser of Proposition~\ref{prop:absorption}: it maximises
$\log p(y\mid\hat w)$ at the posterior-mean coefficient $\hat w$, jointly with the
residual kernel, and so omits the coefficient uncertainty
$\tr(G\Cov(w\mid y)G^\top)$. The range-block residual $y-G\hat w$ is then small, and
the update normalises the total residual by $n$, giving
$\hat\sigma^2\approx\|y-G\hat w\|^2/n\approx\|\Pperp y\|^2/n$. The two normalisations
differ by exactly $\sqrt{(n-r)/n}$. Rescaling accordingly reduces the error from
$3.8\%$ to $\mathbf{0.6\%}$ across all four values of $n$, and from $2.5\%$ to
$\mathbf{0.5\%}$ on the main $n=100$ sweep. The law is then invariant in $n$, as
its statement requires. Table~\ref{tab:absorption-app} reports the uncorrected
numbers.
\end{remark}

\begin{remark}[First-order departure for a fixed in-span covariance]\label{rem:firstorder}
The proof of Proposition~\ref{prop:absorption} gives the maximiser exactly when the
in-span covariances are fitted freely and the range-block bound does not bind. For
a single Gaussian component with a \emph{fixed} in-span covariance, let $b_j$ be
the eigenvalues of $B=U_\parallel^\top GSG^\top U_\parallel$, let $\hat s_j$ be the
average squared range residual along the $j$th eigenvector, and let
$s^\ast=\|\Pperp y\|^2/(n-r)$. Linearising the score around $s^\ast$ gives, to first
order,
\[
  \frac{\hat\sigma^2-s^\ast}{s^\ast}
  =-\frac{s^\ast}{n-r}\sum_{j=1}^r\frac{b_j+s^\ast-\hat s_j}{(b_j+s^\ast)^2},
  \qquad
  \Big|\frac{\hat\sigma^2-s^\ast}{s^\ast}\Big|
  \le\frac{r\,\max_j|b_j+s^\ast-\hat s_j|}{(n-r)\,\lambda_{\min}(B)} .
\]
When the misfit of the range block is at most of the order of the absorbed scale,
the departure is of order $r\hat\sigma^2/\{(n-r)\lambda_{\min}(B)\}$. This is a
first-order statement, not a bound on the exact maximiser.
\end{remark}

\section{Field Coverage in Two Terms}\label{app:coverage}

This appendix separates the basis restriction from absorption in field coverage
(Sections~\ref{sec:experiments} and~\ref{sec:discussion}). Coverage factors into
a basis-restriction term, which survives pinning $\sigma$, and an absorption
term.

\begin{proposition}[Two-term coverage]\label{prop:coverage}
[\textup{\textsc{conditional approximation}}: the likelihood-dominated regime, as stated.]
Assume that $G$ has full column rank ($r=d$) and that the cross-covariance
$\Cov\big(\varphi(t)^\top(\hat w-w^\star),\,\xi_\perp(t)\big)$ is negligible.
In the likelihood-dominated regime ($\sigma$ small, $n\gg d$), the posterior
coefficient covariance is $\Cov(w\mid y)\approx\hat\sigma^2(G^\top G)^{-1}$ and
the posterior mean is $\hat w\approx (G^\top G)^{-1}G^\top y$, independent of the
scalar noise level. The pointwise credible band variance at $t$ is then
$\hat v(t)=\hat\sigma^2 c(t)$ with
$c(t)=\varphi(t)^\top(G^\top G)^{-1}\varphi(t)$, while the true error
$e(t)=\varphi(t)^\top(\hat w-w^\star)+\xi_\perp(t)$ has variance
$\sigma_\parallel^2 c(t)+\tau^2(t)$, $\tau^2(t)=\E[\xi_\perp(t)^2]$. Hence the
nominal-$(1-\alpha)$ pointwise coverage is
\begin{equation}
  \mathrm{cov}(t)
  = 2\,\PhiN\!\left(q_{1-\alpha/2}\sqrt{\frac{\kappa}{1+\gamma(t)}}\right)-1,
  \qquad
  \kappa=\frac{\hat\sigma^2}{\sigma_\parallel^2},
  \qquad
  \gamma(t)=\frac{\tau^2(t)}{\sigma_\parallel^2 c(t)} .
\end{equation}
\end{proposition}

$\gamma$ is the \textbf{basis-restriction term}: it is a property of the truth
and the geometry and survives pinning $\sigma$ at its true value. $\kappa$ is the
\textbf{absorption term}: pinning gives $\kappa_{\text{pinned}}=\sigma^2/\sigma_\parallel^2$.

\begin{corollary}[Free/pinned identity]\label{cor:identity}
[\textup{\textsc{conditional approximation}}, inherited from Proposition~\ref{prop:coverage}.]
Write $\Psi(c):=\left[\PhiN^{-1}\!\left(\tfrac{1+c}{2}\right)\right]^2$. Since
$\gamma$ is identical in both arms,
\begin{equation}
  \boxed{\;\frac{\Psi(\mathrm{cov}_{\text{free}})}{\Psi(\mathrm{cov}_{\text{pinned}})}
  = \frac{\hat\sigma^2}{\sigma^2}\;}
\end{equation}
with nothing fitted.
\end{corollary}

\begin{table}[t]
\centering
\begin{tabular}{rrrr}
\toprule
$\rho$ & $\hat\sigma^2/\sigma^2$ & measured $\Psi$ ratio & error \\
\midrule
1.00 & 0.97 & 0.97 & $+0\%$ \\
1.32 & 1.27 & 1.23 & $-3\%$ \\
2.18 & 2.10 & 1.99 & $-5\%$ \\
4.36 & 4.18 & 3.89 & $-7\%$ \\
10.28 & 9.72 & 8.85 & $-9\%$ \\
33.85 & 31.72 & 27.72 & $-13\%$ \\
58.73 & 55.05 & 47.14 & $-14\%$ \\
136.09 & 127.58 & 102.65 & $-20\%$ \\
306.01 & 288.53 & 206.31 & $-28\%$ \\
\bottomrule
\end{tabular}

\caption{Corollary~\ref{cor:identity} against the sweep, trained to convergence,
with nothing fitted.}
\label{tab:identity}
\end{table}

\begin{remark}[Where the two-term expression stops applying]\label{rem:regime}
Proposition~\ref{prop:coverage} needs $\hat\sigma$ small relative to the prior
scale. Absorption breaks this: $\hat\sigma$ grows as $\sqrt\rho$, reaching
$17\times$ nominal at $\rho=306$. The error of Corollary~\ref{cor:identity} is
below $9\%$ for $\rho\lesssim 10$ and grows with $\hat\sigma/\sigma$, reaching
$-28\%$ at $\rho=306$ (Table~\ref{tab:identity}).

As a prediction of absolute coverage, the pinned-$\sigma$ arm, where the regime
holds, is matched to within $0.0003$ at $\rho=306$. At $\rho=1$ it is
over-predicted by $0.037$, the same as the free arm's $0.038$, so the in-basis
gap is not a regime effect. The median error, $0.003$, hides this shape. The
free-$\sigma$ arm is over-predicted by $0.03$ to $0.057$, growing with
$\hat\sigma$. The basis-restriction term $\gamma$ is therefore matched, and the
discrepancy lies in the absorption term $\kappa$.

The exact coefficient posterior
$C=(G^\top G/\hat\sigma^2+\Sigma_{\text{prior}}^{-1})^{-1}$ changes only the free
arm. The band variance ratio $\varphi^\top C\varphi/\hat\sigma^2 c(t)$ runs from
$0.998$ at $\rho=1$ to $0.840$ at $\rho=306$ with $\sigma$ free, and stays at
$0.999$ with $\sigma$ pinned, where no prediction moves by more than
$4\times10^{-4}$. It removes the growth of the free-arm error, which rises from
$+0.041$ to $+0.067$ under the approximation and is flat at $+0.032$ to $+0.036$
with the exact covariance. The remaining constant offset is not in the band
(Remark~\ref{rem:offset}).
\end{remark}

\begin{remark}[Where the residual offset lives]\label{rem:offset}
The offset comes from in-span leakage of $\xi_\perp$. Splitting $e=P e-\xi_\perp$
on held-out cases, the shrinkage bias contributes a share of $0.00$ of the gap,
the cross-term $-2\E[e_\parallel\xi_\perp]$, which Proposition~\ref{prop:coverage}
assumes negligible, contributes $0.00$, $\tau^2$ is recovered correctly, and excess in-span
\emph{variance} accounts for $1.00$. The model $y=Gw^\star+H\xi_\perp+\varepsilon$
leaks out-of-basis field content into $\range(G)$ and inflates the estimation of
$w$, while $\sigma_\parallel^2 c(t)$ is constant in $\rho$. The in-span second
moment grows from $3.5\times10^{-4}$ to $5.6$ across the sweep against a
predicted $2.98\times10^{-4}$ throughout. On the free-$\sigma$ arm at $\rho\ge2$
the measured error variance exceeds the proposition's error model by a median
factor of $7.3$.

On the free arm the offset is a constant factor on $\hat v/\Var(e)$. Because
$\Psi(\mathrm{cov})=q_{1-\alpha/2}^2\,\hat v/\Var(e)$, the ratio
$\Psi(\text{predicted})/\Psi(\text{measured})$ is the factor by which the error
variance is understated. Across the 31 points on the free-$\sigma$ arm it is
$1.23$ with a coefficient of variation of $0.033$ ($1.22$, CV $0.016$, for
$\rho\ge 2$), whereas the coverage offset itself has CV $0.264$ and the $\Psi$
\emph{difference} CV $1.17$. On the pinned arm the ratio is $1.14$ and decays from
$1.23$ at $\rho=1$ to $1.05$ at $\rho=306$. It therefore does not cancel in the
Corollary's free/pinned ratio, and the residual, $1.26/1.05=1.20$ at $\rho=306$,
is the Corollary's error at the far end of Table~\ref{tab:identity}.

The two shortfalls largely cancel in the coverage ratio. The $\Psi$ ratio
factorises as
$(\hat v_{\mathrm{B}}/v_{\mathrm{actual}})\times(\Var_{\mathrm{actual}}/\Var_{\mathrm{B}})$.
The second factor is the $7.3$ above. The first is $\approx 0.17$ for $\rho\ge2$,
so $0.17\times7.3\approx1.23$, because $\hat\sigma^2c(t)$ is also well below the
model's actual band, which carries the GP residual term as well as the
coefficient posterior. Both shortfalls omit the same leakage-induced in-span
variance, so it enters numerator and denominator in the same direction and
their ratio is far less sensitive to it than either term.

Proposition~\ref{prop:coverage} therefore predicts the coverage curve across two
and a half decades of misspecification up to a constant factor of $1.22$ on the
ratio $\hat v/\Var(e)$, whose source is the omitted in-span leakage of the
out-of-basis field.
\end{remark}

\begin{remark}[Why the collapse begins so early]\label{rem:early}
$\sigma_\parallel^2 c(t)$ is the estimation variance of $d$ coefficients from
$n$ observations at small $\sigma$, so the band is tight because the in-span
problem is easy. $\tau^2$ need not be large to swamp it: at $\rho=1.32$ we
measure $\gamma\approx 4$ already. This is why field coverage falls from $0.851$
to $0.547$ at only $32\%$ excess out-of-basis variance.
\end{remark}

An $n$-sweep ($n\in\{25,50,100,200\}$, 7 values of $\rho$, 5 seeds, both $\sigma$
arms) supports this band-width mechanism. Field coverage at matched $\rho$ is
monotone in $n$ at every point, and the $\rho$ at which coverage falls below any
given level moves monotonically to larger $\rho$ as $n$ falls. The same sweep
gives the finite-$n$ effect of Remark~\ref{rem:breakdown}.
Figure~\ref{fig:sweep}(b) overlays the absolute coverage prediction for both
$\sigma$ arms.

\section{The Misspecification Sweep}\label{app:sweep}

This appendix gives the sweep behind Section~\ref{sec:experiments}: observation
coverage barely moves while field coverage collapses, and pinning $\sigma$
separates the two terms of Appendix~\ref{app:coverage}. We sweep the generator's
residual amplitude over 31 points log-spaced in $\rho$ from $1.00$ to $306.0$,
holding the Bimodal configuration otherwise fixed, with 10 seeds and two arms per
point ($\sigma$ free, $\sigma$ pinned at truth). Here $\rho$ is computed with the
\emph{true} $\sigma$, which a practitioner does not have.

\begin{table}[t]
\centering
\small
\begin{tabular}{rrrrrr}
\toprule
$\rho$ & obs cov$_{90}$ & miscal. & field cov$_{90}$ (free) & field cov$_{90}$ (pinned) & $\hat\sigma/\sigma$ \\
\midrule
1.00 & 0.897 & 0.018 & 0.851 & 0.857 & 0.98 \\
1.32 & 0.890 & 0.018 & 0.547 & 0.501 & 1.13 \\
2.18 & 0.886 & 0.019 & 0.448 & 0.327 & 1.45 \\
4.36 & 0.883 & 0.018 & 0.401 & 0.210 & 2.04 \\
10.28 & 0.878 & 0.016 & 0.377 & 0.131 & 3.12 \\
18.91 & 0.878 & 0.015 & 0.367 & 0.096 & 4.21 \\
33.85 & 0.879 & 0.017 & 0.361 & 0.071 & 5.63 \\
58.73 & 0.879 & 0.020 & 0.355 & 0.054 & 7.42 \\
98.45 & 0.880 & 0.024 & 0.349 & 0.041 & 9.60 \\
136.09 & 0.881 & 0.026 & 0.344 & 0.035 & 11.29 \\
201.02 & 0.881 & 0.029 & 0.337 & 0.029 & 13.74 \\
306.01 & 0.881 & 0.031 & 0.331 & 0.024 & 16.99 \\
\bottomrule
\end{tabular}

\caption{Twelve of 31 sweep points, evaluated on held-out data. Nominal
coverage $0.90$. Across the sweep the widest seed band is $\pm 0.020$ on field
coverage and $\pm 0.068$ on $\hat\sigma/\sigma$.}
\label{tab:sweep}
\end{table}

\paragraph{Observation coverage.}
Observation coverage stays within $0.878$--$0.897$, a range of $0.020$, and is not
monotone: it reaches its minimum at $\rho\approx19$ and recovers slightly. The
miscalibration area behaves the same way, with $0.0185$ at $\rho=1$, a minimum of
$0.0146$ near $\rho=15$, and $0.0309$ at $\rho=306$. Over the same range
field coverage falls by $0.519$, $26$ times more. Observation coverage first drops
below $0.88$ at $\rho=8.3$, where field coverage has already fallen to $0.38$, and
it never falls below $0.875$. When field coverage has halved, at $\rho=3.5$,
observation coverage has moved by $0.014$. A practitioner sees one point on this
curve, and a held-out observation coverage of $0.881$ against a nominal $0.90$ is
indistinguishable from mild under-coverage with benign causes. The sphericity
null, by contrast, is fixed by $(N,p)$ and needs no reference run.

\paragraph{Field coverage and the two terms.}
Field coverage falls from $0.851$ to $0.331$ with $\sigma$ free, and from
$0.857$ to $0.024$ with $\sigma$ pinned, while $\hat\sigma/\sigma$ goes
$0.98\to 16.99$. At $\rho=1.32$, only $32\%$ excess variance outside
$\range(G)$, coverage has already fallen to $0.547$. The pinned arm measures the
basis-restriction term. At $\rho=1.32$ it stands at $0.501$ against $0.857$
in-basis, so the basis restriction alone costs a third of the coverage and most
of the early collapse is not due to $\sigma$. At $\rho=306$ the free arm retains
$0.331$ against the pinned arm's $0.024$, so under this design $93\%$ of the coverage that survives at the far end of the
sweep is attributable to the fitted noise. This is an attribution by comparing the
two arms, not a causal decomposition. The collapse is the
classical oversmoothing effect \citep{knapik2011bayesian}. Its concealment from
observation-space checks, and the partial support of the remaining coverage, are
absorption. The absolute free-minus-pinned difference is bounded by the free-arm
coverage, so beyond $\rho\approx 15$ only the fraction is interpretable. For the
same reason we report no rank correlations between sweep quantities, which are
all monotone in one dial.

\begin{figure}[t]
\centering
\includegraphics[width=\textwidth]{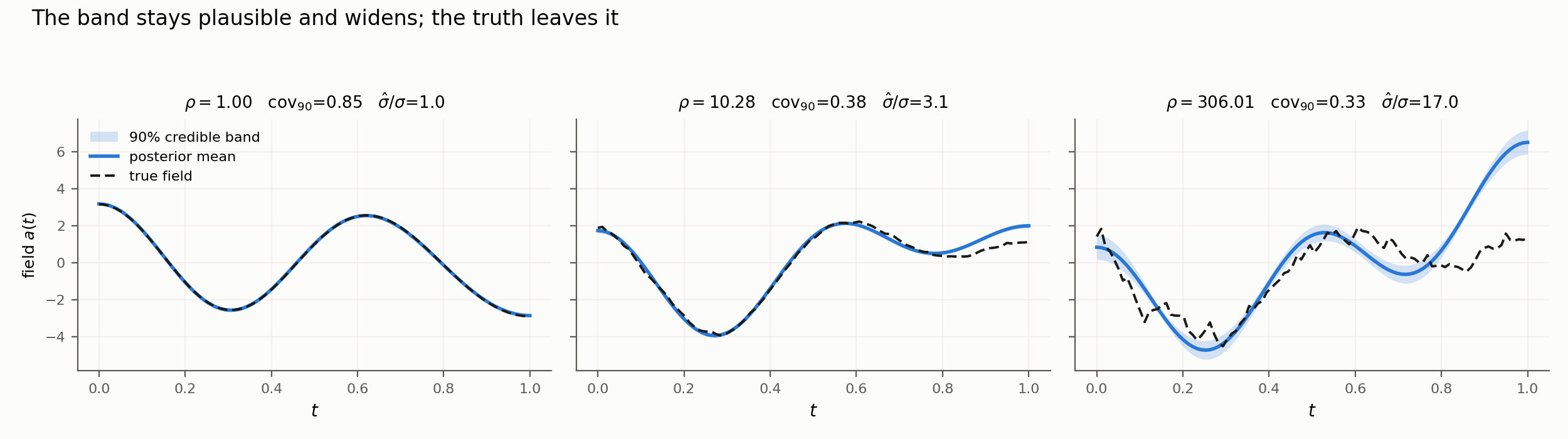}
\caption{Field posteriors as $\rho$ grows, for one benchmark, one seed and one
test case. The fitted noise widens the band, with $\hat\sigma/\sigma$ rising from
$1.0$ to $17.0$, but the truth leaves it. Each panel shows the test case at the
median field RMSE among the 100 held-out cases for its configuration.}
\label{fig:collapse}
\end{figure}

\begin{figure}[t]
\centering
\includegraphics[width=\textwidth]{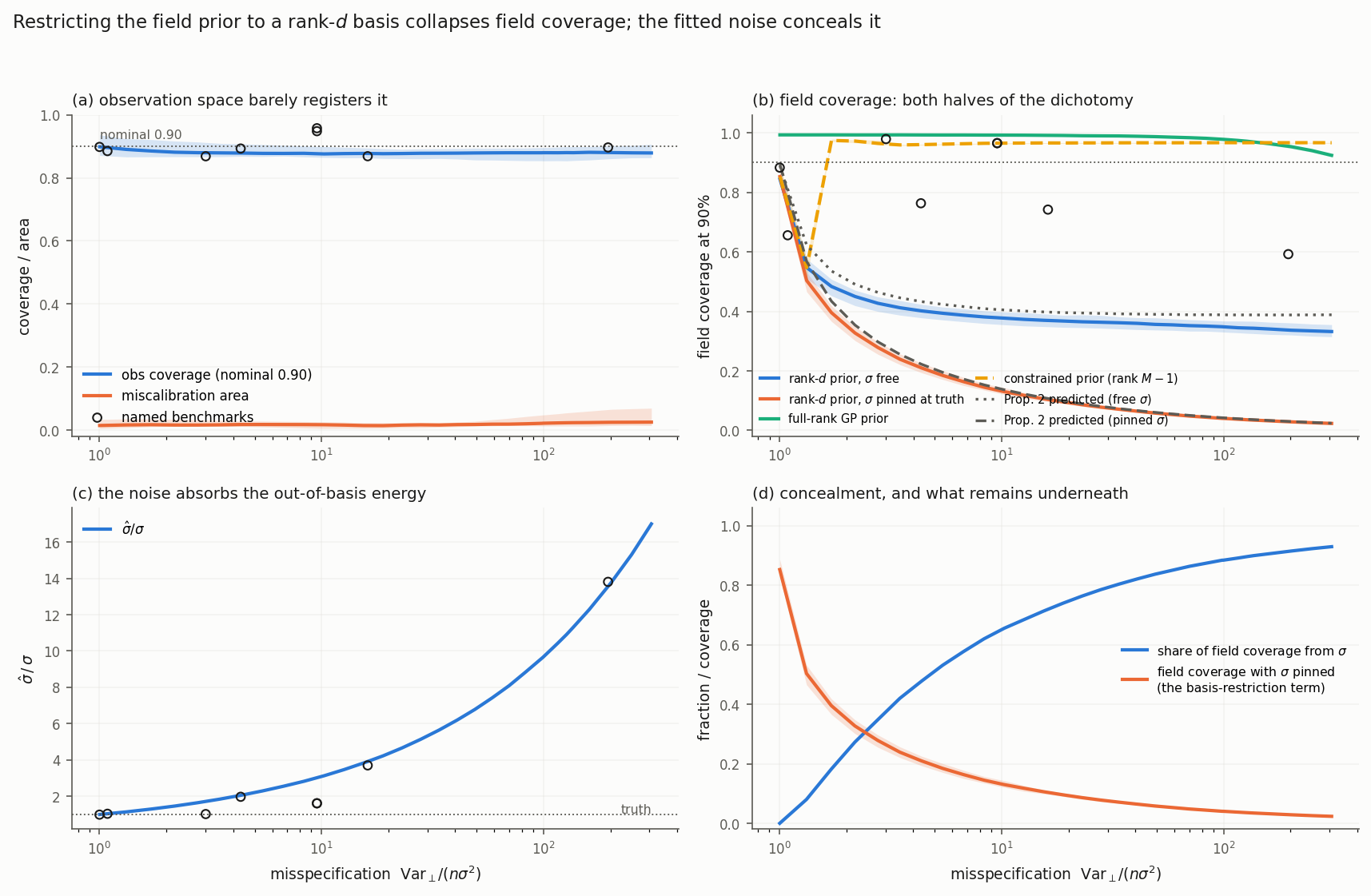}
\caption{The misspecification sweep. (a) Observation coverage and
miscalibration area. (b) Field coverage for the rank-$d$ prior with $\sigma$ free
and pinned, the full-rank GP prior ($0.99$--$0.92$) and the constrained prior at
rank $M-1$, with the prediction of Proposition~\ref{prop:coverage} for both
$\sigma$ arms. (c) $\hat\sigma/\sigma$. (d) The share of field coverage
attributable to $\sigma$, with the pinned curve alongside. The eight named
benchmarks are overlaid as points.}
\label{fig:sweep}
\end{figure}

\section{Proof of Lemma~\ref{lem:nullmean}}\label{app:nullmean}

This appendix proves the exact null mean of Section~\ref{sec:sphericity}.

Let $Q$ be the $N\times N$ Helmert matrix, orthogonal with last row
$N^{-1/2}\mathbf 1^\top$. Under $H_0$, rotating the $N$ complement residuals by $Q$
removes the sample mean and leaves $\nu=N-1$ independent $\Nor_p(0,\sigma^2I_p)$
vectors, the rows of $\tilde Z\in\R^{\nu\times p}$, with $\nu S=\tilde Z^\top\tilde Z$.
$U=p\,\tr(S^2)/(\tr S)^2-1$ is a function of $W=\nu S$ alone and does not depend on
$\sigma$, so take $\sigma=1$. Write $\tilde Z=r\Omega$ with $r=\|\tilde Z\|_F$ and
$\Omega=\tilde Z/\|\tilde Z\|_F$ uniform on the unit sphere of $\R^{\nu p}$ and
independent of $r$. Then $\tr(W)=r^2$ and
$W/\tr(W)=\Omega^\top\Omega$, so $\tr(W)$ is independent of $W/\tr(W)$ at every
$(\nu,p)$. This is the matrix analogue of the independence of a Gaussian vector's
norm and direction, and is why scale-invariant sphericity statistics are pivotal.
Because it does not use the Wishart density, it holds for $p>\nu$, where that
density does not exist. Hence $\E[\tr(W^2)/(\tr W)^2]=\E[\tr W^2]/\E[(\tr W)^2]$.
The numerator is the standard second Wishart moment $\E[W^2]=\nu(\nu+p+1)I_p$
\citep{muirhead1982}, so $\E[\tr W^2]=\nu p(\nu+p+1)$; the denominator follows from
$\tr(W)\sim\chi^2_{\nu p}$, so $\E[(\tr W)^2]=\nu p(\nu p+2)$. Substituting gives
\eqref{eq:nullmean}.

The moments and the independence are classical.
The classical fixed-$p$ approximation $\tfrac{\nu(p+2)}{2}U\to\chi^2_{(p-1)(p+2)/2}$
carries the right degrees of freedom but implies $\E[U]=(p-1)/\nu$, which is wrong by
a factor $(p+2)/p$: $40\%$ at $p=5$ and $2\%$ at $p=96$. At $(p,\nu)=(5,7)$,
simulation gives $0.7575$ against \eqref{eq:nullmean}'s $0.7568$ and the $\chi^2$
approximation's $0.5714$. The null law of $U$ moves with $(N,p)$, its exact mean running
from $47.99$ at $(N,p)=(3,96)$ to $0.121$ at $(800,96)$.

\section{The Test on the Benchmarks}\label{app:practice}

This appendix evaluates the test on the trained benchmark artifacts of
Section~\ref{sec:experiments}. The test is calibrated, detects at the sweep's
first misspecified point, and clause (c) of Proposition~\ref{prop:sphericity}
predicts the leading eigenvalues. We use the full held-out split, $N=800$ and
$p=96$, on the same converged artifacts as the sweep. $\Pperp$ depends only on
$(H,\Phi)$, so nothing is retrained. Appendix~\ref{app:power} covers other
$(N,p)$.

\paragraph{Calibration and detection.}
Against an exact null, with $\xi_\perp$ identically zero, the measured
false-positive rate over 40 trials is $0.025$ at a nominal $0.05$. The sweep's
lowest point is not an exact null: it carries out-of-basis energy of
$1.5\times10^{-6}$ relative to the field, and the test finds it in $2$ seeds of
$10$, a detection rate of $0.20$. At $\rho=1.32$ the detection rate is $1.00$,
and it remains $1.00$ at every larger $\rho$ (Figure~\ref{fig:sphericity}(a)).

\paragraph{Clause (c).}
Proposition~\ref{prop:sphericity}(c) predicts leading eigenvalues at
$\sigma^2(1+(\rho-1)p/k)$. On every operator family here the excess is dominated
by one direction, so $k=1$:

\begin{center}\small
\begin{tabular}{lrrr}
\toprule
benchmark & predicted & measured & error \\
\midrule
Low-SNR & 8.95 & 9.38 & $+4.8\%$ \\
Darcy (1D) & 189.2 & 194.8 & $+3.0\%$ \\
High-dim & 300.7 & 306.4 & $+1.9\%$ \\
Bimodal & 800.9 & 803.4 & $+0.3\%$ \\
Trimodal & 800.9 & 803.4 & $+0.3\%$ \\
Sensors & 237.9 & 215.1 & $-9.6\%$ \\
Deconv & 17863 & 17324 & $-3.0\%$ \\
\bottomrule
\end{tabular}
\end{center}

The median error is $3.0\%$ and the worst $9.6\%$, over the 7 misspecified
benchmarks. The worst is Sensors, where $p=16$ leaves the approximation less room
than at $p=96$. On the sweep, $31.7$ is predicted against $32.4$ measured at
$\rho=1.32$. The second eigenvalue confirms $k=1$. On Low-SNR, the least
concentrated misspecified benchmark with a participation ratio of $6.8$, the ratio
of $\lambda_2$ to the bulk is $1.85$, between the null's own second and leading
ratios of $1.81$ and $1.88$.

\paragraph{Small complement dimension.}
Sensors subsamples to $20$ observations, which leaves $p=16$ with $N\approx800$,
so $p/N=0.02$. Its spectrum has one direction at $215$ times the bulk and the
remaining fifteen spanning $0.76$--$1.53$, consistent with sampling noise about a
flat bulk; detection is $1.00$. Subsampling concentrates the excess, so the test
works at small $p$ for the same reason it works at $p=96$.

\begin{figure}[t]
\centering
\includegraphics[width=\textwidth]{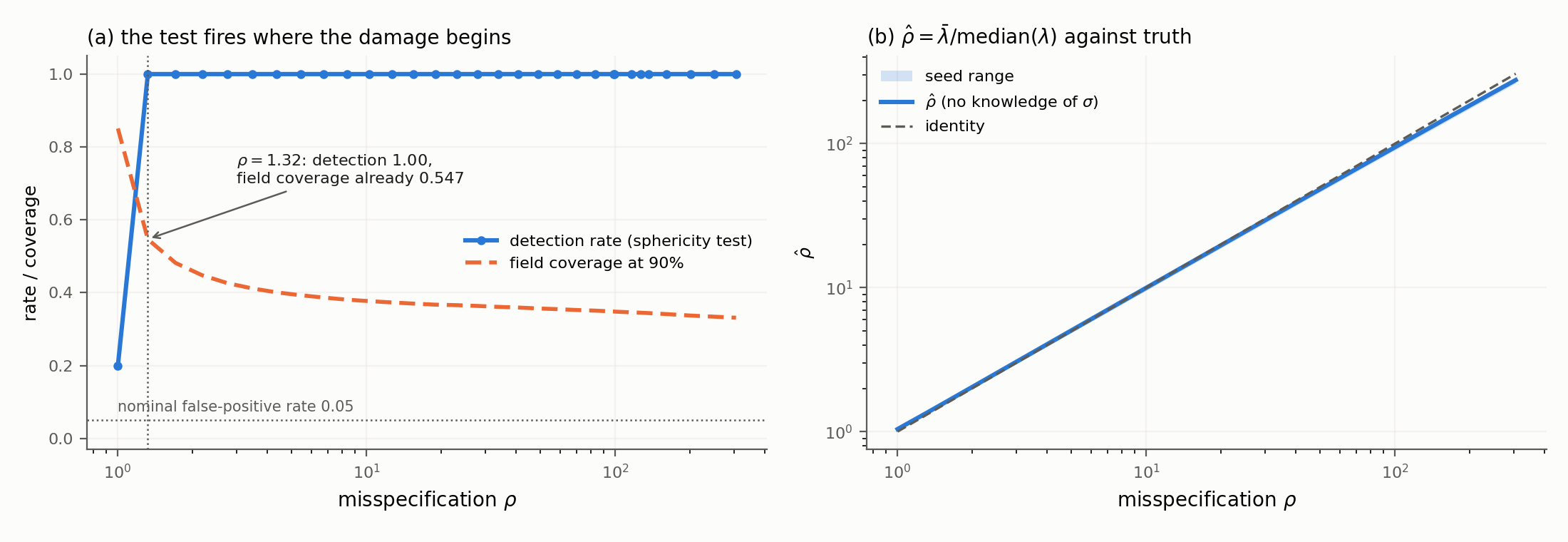}
\caption{(a) Detection rate of the sphericity test against $\rho$, with field
coverage on the same axis. (b) $\hat\rho=\bar\lambda/\mathrm{median}(\lambda)$
against the true ratio, seed ranges shaded, identity dashed. Both are computed
without $\sigma$, at $(N,p)=(800,96)$, over 10 seeds.}
\label{fig:sphericity}
\end{figure}

\section{The Power Study}\label{app:power}

This appendix characterises the test away from the benchmarks: at small $N$,
away from the spiked case, and across $p/N$ including the singular regime. It
supports the Power and Blind spot paragraphs of Section~\ref{sec:experiments}.

No model is trained. The complement residual is drawn from the law
Proposition~\ref{prop:sphericity} assigns it,
\begin{equation}\label{eq:gen}
  z\sim\Nor\!\big(0,\;\sigma^2(I_p+A)\big),\qquad \tr(A)/p=\rho-1,
\end{equation}
with only the \emph{eigenvalues} of $A$ specified. This loses nothing, because
the statistic is a function of the sample covariance spectrum and
$\operatorname{spec}\operatorname{cov}(ZQ)=\operatorname{spec}\operatorname{cov}(Z)$
for orthogonal $Q$. We fix $\sigma=1$, as clause (a) permits: at a fixed seed,
$\sigma\in\{0.01,1,100\}$ give $U$, $\hat\rho$ and $R$ identical to
$4.9\times10^{-16}$ relative. The statistic is imported from the code that runs
on the real artifacts, which reproduces every stored number to a worst relative
disagreement of $2.1\times10^{-15}$ over all $45$ (dataset, seed) pairs. In total
there are $1047$ cells and $240\,600$ replicates.

\paragraph{Calibration.}
The null law of $U$ depends on $(N,p)$: its bootstrapped mean runs from $47.94$ at
$N=3$ (exact value $47.99$, Lemma~\ref{lem:nullmean}) to $0.121$ at $N=800$, a
factor of $400$. Every cell is therefore compared with a
bootstrap at its own $(N,p)$. Over the $98$ exact-null cells, at $59$ distinct
$(N,p)$ pairs and $19\,600$ draws, the false-positive rate is $\mathbf{0.0509}$
against a nominal $0.05$ (binomial standard error $0.0016$). One cell of $98$
exceeds $0.05+3\,\mathrm{s.e.}$, which is what $98$ cells should produce.

\paragraph{Sample size.}
At $p=96$, $k=1$ and $\rho=1.32$, detection exceeds $0.8$ at $\mathbf{N=5}$,
where $p/N=19.2$ and the sample covariance has rank $4$:

\begin{center}\small
\begin{tabular}{lrrrrrrrrr}
\toprule
$N$ & 3 & 4 & \textbf{5} & 6 & 8 & 10 & 15 & 25 & $\ge$50 \\
\midrule
false positives ($\rho=1$) & .050 & .040 & \textbf{.065} & .045 & .070 & .083 & .050 & .050 & .030--.055 \\
detection ($\rho=1.32$) & .495 & .640 & \textbf{.810} & .850 & .950 & .980 & .995 & 1.000 & 1.000 \\
\bottomrule
\end{tabular}
\end{center}

Each cell has $200$ replicates, and $400$ at $N=10$, which two cells cover. The
false-positive rate stays within $0.030$--$0.083$, every cell within about two
standard errors of nominal. By clause (c) the leading eigenvalue sits at
$1+(\rho-1)p\approx31.7$ times the bulk, and one direction at $32\times$ the noise
is visible in a handful of observations.

\paragraph{Structure families and clause (d).}
At $(N,p)=(800,96)$, with every structure normalised to the same
$\tr(A)/p=\rho-1$, detection is $1.000$ for every anisotropic shape: $k$ equal
spikes for $k\le80$, polynomial decay $\lambda_j\propto j^{-\alpha}$ for
$\alpha\in\{0.5,1,2\}$, and exponential decay $\lambda_j\propto e^{-\beta j}$ for
$\beta\in\{0.05,0.2,1\}$, at $\rho=1.32$, $4.36$ and $33.85$ alike. The shapes
separate only at smaller $N$ (Table~\ref{tab:collapse}). Exact isotropy is the
exception clause (d) predicts. Over $98$ isotropic cells at
$\rho\in\{1.05,1.32,4.36,33.85\}$, a factor of $32$ in $\rho$, the detection rate
has median $0.045$, range $[0.015,0.095]$, the nominal rate at every $\rho$. The
calibrated $\hat\rho$ is $0.9999$ and $\Rstd$ is $1.0018$, where clause (d)
predicts exactly $1$.

\paragraph{Dimension and the singular regime.}
Crossing $p\in\{16,48,96,135,250\}$ with $N\in\{25,100,400,1600\}$ spans $p/N$
from $0.01$ to $10$. Of the $40$ misspecified cells, $39$ detect at $1.000$ and
the lowest is $0.985$. The six cells with $p>N$, where the sample covariance is
singular of rank $N-1$, all detect at $1.000$ at $\rho=1.32$, up to
$(p,N)=(250,25)$. The test does not rely on $N>p$: its null is bootstrapped at each $(N,p)$, and the
false-positive rate stays near nominal down to $N=3$ at $p=96$ (see Sample size).
As asymptotic support, the limit $\nu U-p\to\Nor(\E[x^4]-2,4)$, where $\E[x^4]$ is the
fourth moment of a standardised entry, established by
\citet{ledoitwolf2002} for Gaussian data and by \citet{wangyao2013} for general
fourth moment at $p/\nu\to c\in(0,\infty)$, continues to hold as
$p/\nu\to\infty$ \citep{liyao2016}.

\paragraph{Small excess.}
At $\rho=1.05$, a five per cent excess, detection is $1.000$ at every $N\ge100$
with $k=1$, and $0.63$ at $N=25$. At $\rho=1.10$ it is $1.000$ at $N\ge100$ and
$0.965$ at $N=25$. On the benchmark geometries at $\rho=1.05$ it is $0.995$ at
$(p,N)=(16,400)$, $1.000$ at $(96,800)$ and $(135,800)$, and $0.94$ at
$(250,25)$. The sweep did not reach this regime.

\paragraph{Width of the blind spot.}
At $p=96$, the largest fraction of complement directions the excess may occupy
while detection still reaches $0.8$ is:

\begin{center}\small
\begin{tabular}{lrrrrrrrr}
\toprule
$N$ & 20 & 35 & 60 & 100 & 200 & 400 & 800 & 1600 \\
\midrule
$k/p$ at detection $0.8$, $\rho=1.32$ & 0.10 & 0.26 & 0.26 & 0.50 & 0.67 & 0.67 & 0.83 & 0.92 \\
$k/p$ at detection $0.8$, $\rho=4.36$ & --- & --- & 0.83 & 0.83 & 0.92 & 0.98 & 0.98 & 0.99 \\
\bottomrule
\end{tabular}
\end{center}

These are the largest values on our $k$ grid, so each is a lower bound. At
$\rho=1.32$ the test keeps power $0.8$ until the energy is spread over more than
about $90\%$ of the complement at $N=800$, about $70\%$ at $N=200$ and $50\%$ at
$N=100$. Larger $\rho$ pushes the boundary out. The boundary is where $\Rstd$
falls to $1$.

\section{The Auxiliary Statistics $\hat\rho$, $R$ and $\Rstd$}\label{app:aux}

The test runs on $U$. This appendix covers the two auxiliary quantities of
Section~\ref{sec:sphericity}: the magnitude $\hat\rho$, which under-stated the excess
wherever the test fired, and the standardised ratio $\Rstd$, which locates the detection
threshold.

\subsection*{The magnitude $\hat\rho$}

On the benchmarks, over $\rho\in[1,5]$, the median relative error of $\hat\rho$ is
$2.08\%$ and the worst $4.02\%$ (Figure~\ref{fig:sphericity}(b)). That accuracy
holds because the excess occupies one complement direction there. $\hat\rho$
needs $N>p$ and an excess in fewer than half the complement directions; beyond
that, $\operatorname{median}(\lambda)$ stops estimating $\sigma^2$. At
$\rho=1.32$, $N=800$ and $p=96$:

\begin{center}\small
\begin{tabular}{lrrrrrrr}
\toprule
$k$ & 1 & 5 & 10 & 25 & \textbf{48} & 64 & 96 \\
\midrule
detection & 1.000 & 1.000 & 1.000 & 1.000 & \textbf{1.000} & 1.000 & 0.050 \\
error of $\hat\rho-1$ & $-1.3\%$ & $-10.1\%$ & $-19.3\%$ & $-46.2\%$ & $\mathbf{-87.4\%}$ & $-97.3\%$ & $-100\%$ \\
\bottomrule
\end{tabular}
\end{center}

At $k\le25$ the median lands in the clean bulk. At $k=48=p/2$ it lands in the
inflated half, and the test still fires with certainty while $\hat\rho$ recovers
almost none of the excess. On the benchmarks the fraction of complement
eigenvalues above $1.05\times$ the bulk is flat at $0.462$ across the sweep,
inside $k<p/2$ by under four percentage points.

What holds everywhere is one-sidedness. Over \textbf{all $527$ cells of the study in
which the test fires} at detection $\ge0.8$, $\hat\rho-1\le1.10\,(\rho-1)$ in
$\mathbf{100.0\%}$ of them, with median error $-19.1\%$. $\hat\rho-1$ therefore
under-stated $\rho-1$ in every one of the $527$ cells: it is tight when the
misspecification is concentrated and vacuous as it spreads.

$\hat\rho$ drifts to $-9.6\%$ at $\rho=306$. The excess does not spread over more
directions as $\rho$ grows, since the $0.462$ fraction is flat. What moves is the
bulk estimate, $\mathrm{median}(\lambda)/\sigma^2$ rising from $0.962$ to $1.093$,
and $1/1.093 - 1 = -8.5\%$ accounts for almost all of the drift.

On the spiked surface of Figure~\ref{fig:power}, accuracy and detection run close
together. Accuracy is the relative error of the \emph{excess}, $\hat\rho-1$
against $\rho-1$; relative to $\rho$, an estimate of exactly $1$ would score a
near-zero error wherever $\rho\approx1$. Of the $69$ cells with detection
$\ge0.95$, $67$ recover the excess to within $25\%$; of the $31$ with detection
$<0.5$, only $4$ do. Both exceptions are the weakest signal the test can see: at
$(\rho-1)p/k=1$ detection is $0.978$ at $N=800$ while the excess error is
$-37\%$. In Figure~\ref{fig:power} the simulated cells sit at the population excess
$(\rho-1)p/k$, and the benchmarks at their measured excess above the null,
$R-R_0(N,p)$, as a proxy. The sweep point lands at $30.09$ against a population $(\rho-1)p/k$ of
$30.72$.

\subsection*{The standardised ratio $\Rstd$}\label{app:rstd}

$R$ is not comparable across $(N,p)$, because the bulk is estimated and the
sample spectrum spreads as $N$ falls. Under sphericity $R$ has median $1.56$ at
$(N,p)=(1600,96)$, $1.86$ at $(800,96)$, $2.36$ at $(400,96)$ and $5.70$ at
$(100,96)$. Standardising is necessary. At $(N,p)=(800,96)$, unmodelled energy
spread over $48$ of the $96$ complement directions gives $R=2.14$ and detection
$1.000$, while energy spread over all $96$ gives $R=1.86$ and detection $0.040$.
The raw ratio moves by $15\%$ between the two, and $R_0(800,96)=1.86$ is what
distinguishes them.

Table~\ref{tab:collapse} gives each candidate variable at the $0.8$-detection
crossing, over the $21$ series whose crossing lies in the domain $N>p/2$: $14$
misspecification shapes at $p=96$ and $\rho\in\{1.05,1.32\}$, and $p$ from $16$ to
$250$ at fixed excess. Only $\Rstd$ collapses. The collapse survives being split
along either axis: CV $0.068$ over the $11$ series that vary the shape at fixed
$p$, and CV $0.113$ over the $10$ that vary $p$ from $16$ to $250$.

\begin{table}[t]
\caption{Each candidate variable at the $0.8$-detection crossing, over $21$ series
with $200$ replicates per cell. The raw ratio's null level moves with $(N,p)$, and
the population dispersion carries a factor of $p$ that $\Rstd$ does not.}
\label{tab:collapse}
\centering\small
\begin{tabular}{lrrrr}
\toprule
candidate variable & median & range & max/min & CV \\
\midrule
$R=\lambda_1/\operatorname{median}(\lambda)$, raw & 4.90 & 1.88--$5.5\!\times\!10^{4}$ & $2.9\!\times\!10^{4}$ & 3.44 \\
$N\,U_{\mathrm{pop}}/p$, whole-spectrum dispersion & 0.060 & 0.020--0.650 & 31.7 & 1.39 \\
\textbf{$\Rstd=R/R_0(N,p)$} & \textbf{1.193} & \textbf{1.026--1.688} & \textbf{1.65} & \textbf{0.115} \\
\bottomrule
\end{tabular}

\end{table}

The threshold is a band that moves with $(N,p)$:

\begin{center}\small
\begin{tabular}{lccc}
\toprule
 & $N>2p$ & $p/2<N\le p$ & $N\le p/2$ \\
\midrule
$\Rstd$ at detection $0.8$ & $1.17$ & $1.26$ & $1.52$ \\
\bottomrule
\end{tabular}
\end{center}

The CV across series is $0.072$ for $N>2p$, $0.067$ for $p/2<N\le p$ and $0.296$
for $N\le p/2$. A value of $\Rstd$ above $1.3$ indicated power of at least $0.8$
anywhere in $N>p/2$; between $1.0$ and $1.3$ the answer depends on $N/p$. $R_0$
has Monte-Carlo error below $0.5\%$ of $R_0$ at every $N$ from $1600$ down to
$50$. What moves the band is the bulk: $\operatorname{median}(\lambda)/\bar\lambda$
falls from $0.979$ at $N=16.7p$ to $0.646$ at $N=p$ and $0.213$ at $N=0.52p$.
Below $N=p/2$ more than half the sample spectrum is exactly zero and $\Rstd$
should not be read: at $p=96$, $R_0$ is $25.6$ at $N=50$ and $3.9\times10^{15}$ at
$N=48$.

Reading $\Rstd$ off each benchmark's own complement spectrum places the suite on
one scale: In-basis $1.01$ and Darcy 2D $1.01$, the two correctly specified
problems, then Low-SNR $5.1$, the sweep's $\rho=1.32$ point $17.2$, Sensors
$169$, Darcy 1D $105$, High-dim $180$, Bimodal and Trimodal $433$, and Deconv
$9363$. Every misspecified benchmark is one to four orders of magnitude above the
band, which is why detection is $1.00$ on all of them and why the sweep alone
could not locate the boundary.

\section{Asymptotics}\label{app:asymptotics}

This appendix gives the limit theory behind the power law of
Section~\ref{sec:experiments} and the link from field error to detection used in
Section~\ref{sec:real}. Limit theory accounts for the null entirely, for the power
law almost entirely for Gaussian data (Theorem~\ref{thm:prop}), and for the spread
of the $\Rstd$ threshold but not its closed form.

\begin{proposition}[Consistency in $N$ at fixed $p$]\label{prop:consistency}
[\textup{\textsc{asymptotic}} in $N$ at fixed $p$; \emph{not} a finite-sample statement.]
Fix $p\ge2$ and let $\Sigma=\sigma^2(I_p+A)$ with $A\succeq0$ \emph{not} proportional to
$I_p$. Let $U_N$ be computed from $N$ i.i.d.\ $\Nor(0,\Sigma)$ complement residuals and let
$u_{\alpha,N}$ be the $(1-\alpha)$ quantile of its null law at $(N,p)$. Then
$\Pr(U_N>u_{\alpha,N})\to1$ as $N\to\infty$.
\end{proposition}

\begin{proof}
(i) $U\ge0$ and $\E[U\mid H_0]=(p-1)(p+2)/(\nu p+2)\to0$ by Lemma~\ref{lem:nullmean}, so
Markov's inequality gives $\Pr(U>\varepsilon\mid H_0)\to0$ for every $\varepsilon>0$ and
hence $u_{\alpha,N}\to0$. (ii) $S_N\to\Sigma$ a.s., and $U$ is continuous in the spectrum
wherever $\bar\lambda>0$, so $U_N\to U_{\mathrm{pop}}(\Sigma)$ a.s., where
\begin{equation}\label{eq:upop}
  U_{\mathrm{pop}}(\Sigma) := \frac{p\,\tr(\Sigma^2)}{(\tr\Sigma)^2}-1
\end{equation}
is $U$ evaluated on the population spectrum. (iii) $U_{\mathrm{pop}}(\Sigma)=0$ if and only
if all eigenvalues of $\Sigma$ coincide, i.e.\ $\Sigma\propto I_p$; by hypothesis
$U_{\mathrm{pop}}>0$. Combining, $\Pr(U_N>u_{\alpha,N})\to1$.
\end{proof}

The exclusion in Proposition~\ref{prop:consistency}, $A$ not proportional to
$I_p$, is the blind spot of Proposition~\ref{prop:sphericity}(d). Finite-sample
power, consistency and the blind spot therefore share one condition:
$U_{\mathrm{pop}}>0$, that is $\Sigma\not\propto I_p$, that is the unmodelled
energy is not white in the complement.

\paragraph{The scalar the power depends on.}
Two results give the power of John's test in terms of the population spectrum, and both depend on $\Sigma$ only through the single scalar
$U_{\mathrm{pop}}$ of \eqref{eq:upop}. \citet{wangyao2013} give, for real Gaussian data at
$p/\nu\to y\in(0,\infty)$ against a spiked alternative with finitely many spikes $a_i$ of
multiplicity $n_i$, the signal term $\tfrac{\nu}{p}\sum_i n_i(a_i-1)^2$, which is
$\nu\,U_{\mathrm{pop}}$ to leading order, and hence
\begin{equation}\label{eq:wy}
  \beta \;=\; \Phi\!\left(\tfrac{1}{2}\,\nu\,U_{\mathrm{pop}} - z_\alpha\right).
\end{equation}
\citet{liyao2016} give, for general $\Sigma$ of bounded spectral norm at $p/\nu\to\infty$ with
$\nu^3/p=O(1)$, the signal term $(\theta/\gamma^2-1)\nu$ with $\gamma=\tr(\Sigma)/p$ and
$\theta=\tr(\Sigma^2)/p$, so that $\theta/\gamma^2=1+U_{\mathrm{pop}}$ exactly, and hence
\begin{equation}\label{eq:ly}
  \beta \;=\; \Phi\!\left(\frac{(\nu+1)U_{\mathrm{pop}}-2z_\alpha}{2\,(1+U_{\mathrm{pop}})}\right).
\end{equation}
The two agree to leading order and differ only in the $(1+U_{\mathrm{pop}})$ factors, which
matter when the misspecification is large.

This scalar explains why the candidate variables of Table~\ref{tab:collapse}
failed. $N U_{\mathrm{pop}}/p$ carries a spurious factor of $p$ that neither
\eqref{eq:wy} nor \eqref{eq:ly} contains, and $\nu U_{\mathrm{pop}}$ alone omits the
$(1+U_{\mathrm{pop}})$ correction, so its crossing drifts with $U_{\mathrm{pop}}$
as \eqref{eq:ly} predicts.

\paragraph{Coverage of the cited regimes.}
\citet{wangyao2013} is general in $p/\nu$ but assumes finitely many spikes;
\citet{liyao2016} is general in $\Sigma$ but assumes $p/\nu\to\infty$. Neither covers
general spectral alternatives at proportional $p/\nu$, the case our benchmarks
occupy (Table~\ref{tab:regimes}). Theorem~\ref{thm:prop} covers it for Gaussian
data; the non-Gaussian behaviour of Section~\ref{sec:real} is measured, not
derived.

\begin{table}[htbp]
\caption{The two cited power functions and the region between them. All three give the
power through the scalar $U_{\mathrm{pop}}$; Theorem~\ref{thm:prop} adds the spectral
functional $\kappa$ of \eqref{eq:kappa}, which both cited results drop.}
\label{tab:regimes}
\centering\footnotesize\setlength{\tabcolsep}{3pt}
\begin{tabular}{lcccc}
\toprule
 & \citet{wangyao2013} & \citet{liyao2016} & Theorem~\ref{thm:prop} & our cells \\
\midrule
$p/\nu$ & $\to c\in(0,\infty)$ & $\to\infty$, $\nu^3/p=O(1)$ & $\to y\in(0,\infty)$ & $0.01$ to $31$ \\
spectrum of $\Sigma$ & finitely many spikes & general, bounded norm & general, bounded norm & $1$ to $p$ spikes; poly., exp.\ decay \\
\bottomrule
\end{tabular}
\end{table}

\subsection*{A power theorem at proportional $p/\nu$, Gaussian case}\label{app:propthm}

Throughout, $z_1,\dots,z_N$ are i.i.d.\ $\Nor_p(\mu,\Sigma_p)$ with $\Sigma_p$ real,
$\nu=N-1$, $S=\frac1\nu\sum_i(z_i-\bar z)(z_i-\bar z)^\top$ and
$U=p\,\tr(S^2)/(\tr S)^2-1$. Since $U$ is invariant under $\Sigma_p\mapsto c\,\Sigma_p$,
we normalise: $T_p=p\,\Sigma_p/\tr\Sigma_p$, $\mathcal{L}_p=F^{T_p}$ its spectral distribution,
$m_k(\mathcal{L})=\int t^k\,d\mathcal{L}(t)$, so that $m_1(\mathcal{L}_p)=1$ and
$U_{\mathrm{pop},p}=m_2(\mathcal{L}_p)-1$ is \eqref{eq:upop}. Define
\begin{align}
  \kappa(\mathcal{L})&=\int\big(t^2-m_2(\mathcal{L})\,t\big)^2\,d\mathcal{L}(t)\nonumber\\
            &=m_4-2m_2m_3+m_2^3 ,\label{eq:kappa}\\
  \tau^2(\mathcal{L},y)&=m_2(\mathcal{L})^2+\tfrac{2}{y}\,\kappa(\mathcal{L}).\label{eq:tau}
\end{align}
$\kappa\ge0$, with equality iff $\mathcal{L}$ is supported on $\{0,m_2\}$, that is iff $\Sigma$ is
proportional to a projection. This is not a blind spot: there
$U_{\mathrm{pop}}=p/\rank\Sigma-1>0$, the test has power, and only the $\kappa$
correction vanishes, leaving the $1+U_{\mathrm{pop}}$ scale of \eqref{eq:ly}. In the
complement model $\Sigma=\sigma^2(I_p+A)$ every eigenvalue of $T_p$ is at least
$1/\rho$, so an excess flat on a subspace, $A=c\,P$, has $\kappa>0$, and $\kappa=0$
iff $\Sigma\propto I_p$.

\begin{theorem}[John's $U$ at proportional $p/\nu$, Gaussian case]\label{thm:prop}
[\textup{\textsc{asymptotic}} as $p\to\infty$ with $p/\nu\to y$; Gaussian data only.]
Suppose, as $p\to\infty$,
(H1) $p/\nu\to y\in(0,\infty)$;
(H2) $\sup_{p}\,\|T_p\|<\infty$ (spectral norm);
(H3) $\mathcal{L}_p\to \mathcal{L}$ weakly, with $\mathcal{L}$ a proper c.d.f.
Then
\begin{equation}\label{eq:propclt}
  \tfrac \nu2\big(U-U_{\mathrm{pop},p}-\tfrac p\nu\big)-\tfrac12\big(1+U_{\mathrm{pop},p}\big)
  \;\Rightarrow\;\Nor\big(0,\tau^2(\mathcal{L},y)\big).
\end{equation}
\end{theorem}

The proof assembles classical ingredients.
In the complement model, $T_p=(I_p+A_p)/\rho_p$ and (H2) is $\sup_p\|A_p\|<\infty$,
which excludes only spikes that grow with $p$. Under (H2) the measures $\mathcal{L}_p$ live on
one compact interval, so (H3) asks only that the limit exist. At $\Sigma_p\propto I_p$, $\tau=1$ and \eqref{eq:propclt} reads
$\nu U-p\Rightarrow\Nor(1,4)$, the Gaussian case $\E[x^4]=3$ of the cited null law
\citep{ledoitwolf2002}.

\begin{lemma}[Exact trace moments]\label{lem:tracemoments}
[\textup{\textsc{exact}} at every $(\nu,p)$, $p>\nu$ included; Gaussian.] Let
$W=\Sigma^{1/2}XX^\top\Sigma^{1/2}$ with $X$ a $p\times \nu$ matrix of i.i.d.\ $\Nor(0,1)$
entries, and $T_k=\tr\Sigma^k$. Then
\begin{align*}
  \E\tr W&=\nu T_1, \qquad \Var\tr W=2\nu T_2,\\
  \E\tr W^2&=\nu(\nu+1)T_2+\nu T_1^2,\\
  \Cov(\tr W,\tr W^2)&=4\nu\big((\nu+1)T_3+T_1T_2\big),\\
  \Var\tr W^2&=4\nu\big((2\nu^2+5\nu+5)T_4\\
  &\quad+4(\nu+1)T_1T_3+(\nu+1)T_2^2\\
  &\quad+2T_1^2T_2\big).
\end{align*}
With $T_k=p\,m_k$ and $p/\nu\to y$, the fourth central moments of $\tr(W)/\nu$ and $\tr(W^2)/\nu^2$
converge, each to three times the square of its limiting variance.
\end{lemma}

\begin{proof}
Isserlis' theorem: $\E\prod\tr(W^{m})$ is a sum over perfect matchings of the entries of
$X$; a matching forces equal row labels within each pair, and each connected cycle of
column labels, alternating matched pairs and factors of $\Sigma$, contributes
$\tr\Sigma^{\text{length}}$. Each row class contributes $\nu$. The enumeration is mechanical
and is carried out exactly in \texttt{experiments/theory/n3\_variance\_derivation.py}
(second moments) and \texttt{n3\_fourth\_moments.py} (fourth moments, $15!!$ matchings for
$\E(\tr W^2)^4$). At $\Sigma=I_p$ the first two lines reduce to the moments used in
Lemma~\ref{lem:nullmean}: $\E\tr W^2=\nu p(\nu+p+1)$.
\end{proof}

\begin{proof}[Proof of Theorem~\ref{thm:prop}]
\emph{Step 0: scale.} $U$ is scale-invariant, so take $\Sigma_p=T_p$, $\tr T_p=p$.

\emph{Step 1: the centring bridge (Gaussian only).} Let $Q$ be an $N\times N$ orthogonal
matrix whose last row is $N^{-1/2}\mathbf 1^\top$, and $w_1,\dots,w_N$ the columns of
$[z_1,\dots,z_N]Q^\top$. For Gaussian $z_i$ the $w_i$ are independent, $w_1,\dots,w_\nu$
are $\Nor_p(0,T_p)$, and $\sum_i(z_i-\bar z)(z_i-\bar z)^\top=\sum_{i\le \nu}w_iw_i^\top$.
Hence $S$ has the law of the \emph{non-centred}
$B_\nu=\frac1\nu T_p^{1/2}X_\nu X_\nu^\top T_p^{1/2}$ with $X_\nu$ a $p\times \nu$ matrix of i.i.d.\
$\Nor(0,1)$ entries. For non-Gaussian $z_i$ the $w_i$ are uncorrelated but not independent,
and this step fails.

\emph{Step 2: asymptotic normality.} We apply \citet[Theorem~9.10]{baisilverstein2010} to
$B_\nu$ with $T_\nu=T_p$ and $f_1(x)=x$, $f_2(x)=x^2$, and check its hypotheses in turn.
(9.7.2): $\frac1{\nu p}\sum_{ij}\E|x_{ij}|^4I(|x_{ij}|\ge\sqrt \nu\eta)
=\E\,x^4I(|x|\ge\sqrt \nu\eta)\to0$ for every $\eta>0$, by dominated convergence for
$x\sim\Nor(0,1)$.
(a): the $x_{ij}$ are independent with $\E x_{ij}=0$, $\E|x_{ij}|^2=1$,
$\E|x_{ij}|^4=3$, and $p/\nu\to y$ by (H1).
(b): $T_p$ is nonrandom, real symmetric and nonnegative definite, bounded in spectral norm
by (H2), with $F^{T_p}\to \mathcal{L}$ proper by (H3).
Analyticity: with $\lambda_{\max}^+=\limsup_p\|T_p\|<\infty$, the interval
\[
  \big[\liminf_p\lambda_{\min}^{T_p}I_{(0,1)}(y)(1-\sqrt y)^2,\ \lambda_{\max}^+(1+\sqrt y)^2\big]
\]
has lower endpoint in $[0,\lambda_{\max}^+]$ when $y<1$ and equal to $0$ when $y\ge1$; in
both cases the interval lies inside $[0,\lambda_{\max}^+(1+\sqrt y)^2]$, and $f_1,f_2$ are
polynomials, hence analytic on the open disc of radius $\lambda_{\max}^+(1+\sqrt y)^2+1$.
The $y\ge1$ case, $p>\nu$, where $S$ is singular, is therefore covered. The indicator would
bind for a function singular at $0$, such as the $\log$ in the likelihood-ratio statistic,
which is in any case undefined for $p>\nu$.
Part (2) of the theorem requires in addition $x_{ij}$ and $T_p$ real and $\E x_{ij}^4=3$,
which hold. It gives deterministic $c_p\in\R^2$ with
$(\tr B_\nu,\tr B_\nu^2)-c_p\Rightarrow G$, a bivariate Gaussian. We use only this
conclusion and not the theorem's formulas for the limiting mean and covariance.

\emph{Step 3: identifying the limit.} Let $(a_p,b_p)=(\tr S-\E\tr S,\ \tr S^2-\E\tr S^2)$.
By Lemma~\ref{lem:tracemoments} their fourth moments are bounded, so $a_p^2$, $b_p^2$ and
$a_pb_p$ are uniformly integrable. Then $\E\tr S-c_{p,1}$ and $\E\tr S^2-c_{p,2}$ are
bounded; along any subsequence where they converge to $d$, $(a_p,b_p)\Rightarrow G-d$ with
$\E(G-d)=\lim\E(a_p,b_p)=0$, so $d=\E G$ on every such subsequence. Hence
$(a_p,b_p)\Rightarrow\Nor(0,V)$, with $V$ the limit of the exact covariances. With
$T_k=p\,m_k$, Lemma~\ref{lem:tracemoments} gives
\begin{align*}
  V_{aa}&=2y\,m_2,\qquad V_{ab}=4\big(y\,m_3+y^2m_2\big),\\
  V_{bb}&=4\big(2y\,m_4+4y^2m_3+y^2m_2^2+2y^3m_2\big),
\end{align*}
with $m_k=m_k(\mathcal{L})$, since (H2) and (H3) give $m_k(\mathcal{L}_p)\to m_k(\mathcal{L})$.

\emph{Step 4: the delta method, with its remainder.} Let
$q_p=\E\tr S^2/p=(1+\frac1\nu)m_2(\mathcal{L}_p)+\frac p\nu$. Then
$U+1=(q_p+b_p/p)(1+a_p/p)^{-2}$ exactly, and
\begin{align*}
  \tfrac \nu2\big(U+1-q_p\big)&=\tfrac{\nu}{2p}\big(b_p-2q_pa_p\big)+R_p,\\
  R_p=\tfrac \nu2\Big[\big(q_p+\tfrac{b_p}p\big)&\Big(\big(1+\tfrac{a_p}p\big)^{-2}-1+\tfrac{2a_p}p\Big)
  -\tfrac{2a_pb_p}{p^2}\Big].
\end{align*}
Since $a_p,b_p=O_P(1)$ and $(1+x)^{-2}-1+2x=O(x^2)$, $R_p=O_P(\nu/p^2)=o_P(1)$. By Step 3
and $q_p\to m_2+y$, the leading term converges to $\Nor(0,\tau^2)$ with
\begin{align*}
  \tau^2&=\tfrac1{4y^2}\big(V_{bb}-4(m_2+y)V_{ab}+4(m_2+y)^2V_{aa}\big)\\
        &=\tfrac1{4y^2}\big(4y^2m_2^2+8y\,\kappa\big)=m_2^2+\tfrac2y\kappa .
\end{align*}
The $y^3$ and $y^2m_3$ terms, which carry the Marchenko--Pastur bulk, cancel identically;
\texttt{n3\_variance\_derivation.py} checks the algebra symbolically. Finally
$U+1-q_p=U-U_{\mathrm{pop},p}-\frac p\nu-\frac1\nu(1+U_{\mathrm{pop},p})$, which is
\eqref{eq:propclt}.
\end{proof}

\begin{corollary}[Power]\label{cor:proppower}
Under (H1)--(H3), let $\beta_p$ be the power of the level-$\alpha$ test that rejects when
$U$ exceeds its exact null quantile at $(\nu,p)$, which the parametric bootstrap estimates.
Let $\tau_p^2=m_2(\mathcal{L}_p)^2+\frac{2\nu}{p}\kappa(\mathcal{L}_p)$ and
\begin{equation}\label{eq:proplaw}
  \hat\beta_p=\Phi\!\left(\frac{\tfrac{\nu+1}{2}\,U_{\mathrm{pop},p}-z_\alpha}{\tau_p}\right).
\end{equation}
(i) \emph{Local alternatives.} If $\frac{\nu+1}2U_{\mathrm{pop},p}\to h\in[0,\infty)$, then
$\mathcal{L}=\delta_1$, $\tau=1$ and $\beta_p\to\Phi(h-z_\alpha)$, for every spectrum satisfying
(H2).
(ii) \emph{Fixed alternatives.} If $U_{\mathrm{pop}}(\mathcal{L})>0$, then $\beta_p\to1$, and
\eqref{eq:propclt} holds with $\tau^2-(1+U_{\mathrm{pop}}(\mathcal{L}))^2=\frac2y\kappa(\mathcal{L})$, which
is strictly positive in the complement model.
(iii) In every case $\beta_p-\hat\beta_p\to0$.
\end{corollary}

\begin{proof}
At $\Sigma_p\propto I_p$, \eqref{eq:propclt} with $\tau=1$ and continuity of $\Phi$ give a
critical value $u_{\alpha,p}$ with $\frac \nu2(u_{\alpha,p}-\frac p\nu)-\frac12\to z_\alpha$.
Under the alternative, by \eqref{eq:propclt},
$\beta_p=\Pr\big(\tau Z>z_\alpha-\frac{\nu+1}2U_{\mathrm{pop},p}+o(1)\big)+o(1)$ with
$Z\sim\Nor(0,1)$. In (i), $U_{\mathrm{pop},p}=m_2(\mathcal{L}_p)-1\to0$. Under (H2) the
$\mathcal{L}_p$ live on one compact interval, so the weak convergence in (H3)
carries the moments: $m_1(\mathcal{L})=1$ and $m_2(\mathcal{L})=1$. Then
$\mathcal{L}$ has variance $m_2(\mathcal{L})-m_1(\mathcal{L})^2=0$, so it is the
point mass at $1$ and $\kappa(\mathcal{L})=0$. In (ii), $\frac{\nu+1}2U_{\mathrm{pop},p}\to\infty$.
For (iii), split any sequence into subsequences on which
$\frac{\nu+1}2U_{\mathrm{pop},p}$ converges in $[0,\infty]$. Where the limit is finite,
(i) applies and $\tau_p\to1$; where it is infinite, $\beta_p\to1$ because the left side
of \eqref{eq:propclt} is tight, and $\hat\beta_p\to1$ because $\tau_p$ is bounded.
\end{proof}

\paragraph{Scope of the theorem.}
The corollary's part (iii) also holds with $\tau_p$ replaced by $1$ or by
$1+U_{\mathrm{pop},p}$: the limiting \emph{power} does not select $\tau_p$. At local
alternatives the $\kappa$ term is $O(1/p)$ and vanishes, since
$\kappa=\frac1p\sum_jx_j^2(x_j-m_2)^2\le\|T_p\|^2m_2\,U_{\mathrm{pop},p}$ and
$U_{\mathrm{pop},p}=O(1/\nu)$ there; at fixed alternatives it is of
order one and is the correct fluctuation scale of $U$ by \eqref{eq:propclt}, but there the
power tends to one. So $\kappa$ is the finite-$p$ correction at fixed alternatives, and the
advantage of \eqref{eq:proplaw} over its $\kappa=0$ versions at finite $(\nu,p)$ is
established by the surface below, not by the theorem.

Both cited laws are special cases. For one fixed spike $a$, $U_{\mathrm{pop},p}=(a-1)^2/p$
to leading order, and (i) gives $\Phi\big(\nu(a-1)^2/(2p)-z_\alpha\big)$, which is
\eqref{eq:wy} and \citet[eqs.~4.7--4.8]{wangyao2013}; part (i) extends it from finitely
many spikes to every bounded spectrum. The $\kappa$ term for that spike is
$\frac{2\nu}{p}\kappa_p\approx2\nu a^2(a-1)^2/p^2$, which vanishes in their limit but is about
$10a^2/p$ at the power transition: at $a=4$, $p=96$ the law crosses $0.8$ at $N=72$ with
$\tau_p=1.73$ rather than $1$. Letting $y\to\infty$ formally in \eqref{eq:proplaw} sends $\frac2y\kappa\to0$ and
$\tau\to1+U_{\mathrm{pop}}$, which is \eqref{eq:ly}; the theorem itself needs $y<\infty$
and does not prove that limit.

\paragraph{Accuracy at finite $(\nu,p)$.}
We evaluate three laws on all $949$ misspecified cells of the surface---$k=1$ to $p$ equal
spikes, polynomial decay $\alpha\in\{0.5,1,2\}$, exponential decay
$\beta\in\{0.05,0.2,1\}$, $p\in[16,250]$, $N\in[3,3200]$, $\rho\in[1.05,306]$:
\eqref{eq:ly}; \eqref{eq:proplaw}; and \eqref{eq:proplaw} with a third-order term. The third-order term is a finite-sample refinement, not part of the theorem. It replaces the normal by a shifted gamma matched to three moments:
the null's standard deviation and skewness exactly at each $(\nu,p)$, from the Wick
enumeration of $\E(\tr W^2)^k$, $k\le3$, with the independence of Lemma~\ref{lem:nullmean};
and the alternative's mean shift $\frac{\nu+1}2U_{\mathrm{pop}}$, scale $\tau_p$ and third
cumulant from the second-order delta method. Mean absolute errors in the detection rate:

\begin{center}\small\setlength{\tabcolsep}{4pt}
\begin{tabular}{lrrr}
\toprule
cells & \eqref{eq:ly} & \eqref{eq:proplaw} & $+$3rd order \\
\midrule
all $949$                         & 0.019 & 0.016 & 0.011 \\
transition band$^*$ & 0.044 & 0.035 & 0.025 \\
$\nu\ge10$                          & 0.013 & 0.011 & 0.008 \\
$p/\nu<0.1$                         & 0.004 & 0.004 & 0.003 \\
$0.1\le p/\nu<2$                    & 0.011 & 0.009 & 0.008 \\
$p/\nu\ge2$                         & 0.036 & 0.032 & 0.019 \\
few spikes, $k\le5$               & 0.019 & 0.015 & 0.009 \\
many spikes, $5<k<p$              & 0.014 & 0.013 & 0.011 \\
poly.\ or exp.\ decay   & 0.024 & 0.021 & 0.014 \\
\bottomrule
\end{tabular}
\end{center}

$^*$Measured detection in $(0.05,0.95)$. These averages involve no interpolation. The median bias
is $0.000$ for all three. In the transition band at $\nu\ge10$ ($259$ cells) the third-order
error is $0.021$, against $0.019$ expected from binomial noise alone at the surface's
replicate counts if the law were exact: the surface cannot resolve a difference at this
scale. The largest miss at $\nu\ge10$ is a real difference. At $p=96$, $N=100$, polynomial decay $0.5$, the
surface reads $0.315$ and the law $0.411$; re-simulated with $20{,}000$ replicates the cell
is $0.377$ (s.e.\ $0.003$). The surface value was low by noise, and the law is also high,
by $0.035$. Every law fails at $\nu\le4$, and there all are optimistic: at $N=3$, $p=96$
with one spike at $\rho=1.32$ the three predict $0.85$, $0.73$ and $0.67$ against a
measured $0.50$. We state $\nu\ge10$ as the usable range.

Comparing a law inverted exactly with a measured sample size that is log-linear
interpolation between grid points biases the comparison. We therefore interpolate each law
on the same grid. The sample size at power $0.8$ for one spike at $p=96$ is then:

\begin{center}\small
\begin{tabular}{lrrrr}
\toprule
$(\rho-1)p/k$ & measured & \eqref{eq:ly} & \eqref{eq:proplaw} & $+$ third order \\
\midrule
$1$  & 564.6 & 545.8 & 568.6 & 575.4 \\
$3$  & 77.3  & 64.4  & 74.7  & 76.4 \\
$10$ & 14.9  & $\le10$ & 13.5 & 14.6 \\
\bottomrule
\end{tabular}
\end{center}

\eqref{eq:ly} is optimistic, by $3\%$ at $(\rho-1)p/k=1$. At $0.3$, \eqref{eq:proplaw} puts the crossing at $N=5552$ and the power
at $N=1600$ at $0.19$, consistent with no detection on the grid.

Proposition~\ref{prop:sphericity}(d) is the law's own boundary case rather than a separate
exception: $U_{\mathrm{pop}}=0$ exactly when $\Sigma\propto I_p$, so \eqref{eq:ly} returns
$\alpha$ and can return nothing else. Over the $98$ isotropic cells
$\max U_{\mathrm{pop}}=4.9\times10^{-32}$ and the measured median detection is $0.045$.

\paragraph{The collapse of $\Rstd$.}
The law's crossing must be compared with the measured one like for like: the same
series, the same $N>p/2$ rule and the same interpolation, with only the variable that
locates the crossing changed. Otherwise a difference in CV can be mistaken for Monte-Carlo
error. The measured crossing is $1.193$ with CV $0.115$ over $21$ series, and the laws give, over
$22$ series: \eqref{eq:ly} $1.178$, CV $0.087$; \eqref{eq:proplaw} $1.184$, CV $0.110$;
with the third order $1.185$ $[1.019,1.650]$, CV $\mathbf{0.114}$. Paired series by
series, the third-order crossing's median ratio to the measured one is $1.000$, and it
falls below the measured value in $48\%$ of series, against $95\%$ for \eqref{eq:ly}.

The law reproduces the measured spread, so the CV of $0.115$ is real dependence of the
threshold on $N/p$, not Monte-Carlo error. This explains the band of
Appendix~\ref{app:rstd}; the reported threshold remains the measured one.

\emph{Deriving} the band in closed form is open, for a structural reason. $U$ is a linear spectral statistic, governed by
\citet[Theorem~9.10]{baisilverstein2010}, whereas $\Rstd$ is a functional of the
\emph{largest} eigenvalue, governed by Tracy--Widom fluctuations and the BBP transition;
\citet{onatski2013} show the two come apart, in that Tracy--Widom-type tests have
asymptotically trivial power in a contiguity region where the eigenvalue likelihood ratio
does not. A derivation would need the joint behaviour of a linear spectral statistic and
the extreme eigenvalue.

\subsection*{From field error to detection}\label{app:fielderror}

Theorem~\ref{thm:prop} gives power through $U_{\mathrm{pop}}$, a property of the complement
covariance. Two exact identities connect it to the field error.

\begin{proposition}[From field error to $U_{\mathrm{pop}}$]\label{prop:fielderror}
[\textup{\textsc{exact}}; any law of $\xi_\perp$ with finite second moments.] Let the out-of-span
field error $\xi_\perp$ have covariance $K$ in field space, independent of
$\varepsilon\sim\Nor(0,\sigma^2I_n)$, so that $\Cov(z)=\sigma^2(I_p+A)$ with
$\sigma^2A=U_\perp^\top HKH^\top U_\perp$. Write $\mathrm{PR}=(\tr A)^2/\tr(A^2)$ for the
participation ratio of $A$,
\[
  \ell=\frac{\tr(P_\parallel HKH^\top)}{\tr(HKH^\top)},\qquad
  g=\frac{\tr(HKH^\top)}{\tr K}
\]
for the energy-weighted leak into $\range(G)$ and the operator's energy gain on the error.
Then
\begin{align}
  \rho-1&=(1-\ell)\,g\,\frac{\tr K}{p\,\sigma^2},\label{eq:bridge}\\
  U_{\mathrm{pop}}&=\Big(\frac{\rho-1}{\rho}\Big)^2\Big(\frac{p}{\mathrm{PR}}-1\Big).\label{eq:upoppr}
\end{align}
\end{proposition}

\begin{proof}
$\tr A=\tr(U_\perp^\top HKH^\top U_\perp)/\sigma^2=\tr(P_\perp HKH^\top)/\sigma^2
=(1-\ell)\tr(HKH^\top)/\sigma^2$, and $\rho-1=\tr A/p$. For \eqref{eq:upoppr}, the
eigenvalues of $I_p+A$ are $1+a_j$ with mean $\rho$, so
$U_{\mathrm{pop}}=\frac1p\sum_j(a_j-\bar a)^2/\rho^2
=\big(\tr(A^2)/p-(\rho-1)^2\big)/\rho^2$, and $\tr(A^2)/p=(\rho-1)^2p/\mathrm{PR}$.
\end{proof}

On GEBCO both identities hold to rounding in every one of the $36$ cells, the first to
$6.7\times10^{-16}$ and the second to $1.1\times10^{-15}$, with $K$ built from the
evaluation fields and $\tr K/M$ the squared in-span floor.

\paragraph{When the test fires.}
Together with Corollary~\ref{cor:proppower}, the proposition says when the test fires. To
first order, power $\beta$ at level $\alpha$ needs
$\frac{n+1}2U_{\mathrm{pop}}\ge z_\alpha+z_\beta\tau$. Since $((\rho-1)/\rho)^2<1$, this
has two parts. First, the complement excess must be anisotropic enough,
$p/\mathrm{PR}-1>2(z_\alpha+z_\beta\tau)/(n+1)$, and no amount of field error of the same
shape substitutes for that. Second, given the anisotropy, $\rho$ must exceed the $\rho^*$ at which
\eqref{eq:upoppr} reaches the threshold, and by \eqref{eq:bridge} that holds as soon as
\[
  \tr K\;\ge\;\frac{p\,\sigma^2(\rho^*-1)}{(1-\ell)\,g}.
\]
So the test fires when the field error is large enough, \emph{provided} the operator
carries it into the observations ($g$ bounded below) and the fit does not absorb it ($\ell$
bounded away from one).

There is no bound in the field error alone. The worst-case gain
$\min\{\|Hv\|^2/\|v\|^2: v\in\range(K)\}$ is zero here, and necessarily so: $H$ is
$196\times1024$ with a null space of dimension at least $828$, and $\range(K)$ has
dimension $1024-d$, so the two subspaces intersect. What exists is the exact factorisation
\eqref{eq:bridge}, whose two operator factors are computable from $H$, $\Phi$ and any
assumed error covariance. Computed from the basis split rather than the evaluation fields,
with the nominal $\sigma$, \eqref{eq:upoppr} and the local power law predict detection at
$N=100$ with mean absolute error $0.017$ and the right fire/no-fire call in $35$ of $36$
cells. The prediction is conservative by construction: $\Phi$ is fitted to the basis
split, so $\rho-1$ there is a median $0.85$ of its evaluation value.

\paragraph{Leakage across $d$.} As the basis grows the in-span floor
falls from $98$ to $24$ over $d\in\{4,\dots,96\}$, identically under all three operators. A
growing share of what survives leaks into $\range(G)$, where the coefficients absorb it. The
leaked fraction $\|\Ppar H\xi_\perp\|^2/\|H\xi_\perp\|^2$ rises with $d$ under gravity, from
$0.050$ to $0.336$, and under pointwise subsampling, from $0.006$ to $0.239$. $\Upop$ falls
monotonically in $d$ under all three. Field error reaches the test only through the
operator's gain on it and this leak, so large field error need not be detected.

\paragraph{The converse.}
A lower bound on $\rho-1$ gives ``large field error $\Rightarrow$ detection''. It does not
give ``no detection $\Rightarrow$ small field error''. Error in
$\operatorname{null}(H)$ leaves the law of $y$ unchanged. Error that $H$ maps into
$\range(G)$ is absorbed by the coefficients, and under an unconstrained coefficient law it
is indistinguishable from signal. Both
act through \eqref{eq:bridge}, through $g$ and $\ell$ respectively. GEBCO measures both.
Along the gravity sweep at $\sigma=1$~mGal, $\rho-1$ falls by a factor of $598$ from
$d=4$ to $d=96$, and \eqref{eq:bridge} splits that exactly into four factors:
\begin{itemize}
  \item field-error energy $16.3$ (the floor falls from $98.4$ to $24.4$~m);
  \item gain $49.4$, because the low-pass operator barely sees the high-wavenumber error
        that a large basis leaves;
  \item $1-\ell$ only $1.43$, as the leak rises from $0.050$ to $0.336$;
  \item and $p$ from $192$ to $100$, a factor $0.52$ the other way.
\end{itemize}
At $d=96$ that cell is harmed, with the field error at $24.4$~m, and detected in only
$0.20$ of runs at $N=100$: it is the single harmed cell the test misses there. Pointwise
sampling at the same $d$ and the same field error has $2.1\times10^4$ times the gain and
fires at $1.00$. Under the gradient operator the leak spans $0.100$ to $0.241$ over the
same range but falls to $0.051$ at $d=8$ before rising. The leak dose-response is therefore
real but secondary. What mainly
defeats the converse on this problem is the gain, a property of the instrument rather than
of the model, and one that no test on the data could recover.

\section{Real Topography: Design and Further Results}\label{app:real}

This appendix gives the design of the GEBCO study and the detail behind
Section~\ref{sec:real}.

\paragraph{Pre-registration.} The design, the split, the operator constants, the
grid over $d$ and $\sigma$, and the three admissible outcomes were fixed and
committed before any GEBCO file was read. The registered outcomes were: (A) the
complement excess concentrates in few directions; (B) it is close to isotropic,
in which case the diagnostic is blind to real misspecification and the paper
reports that; (C) the answer depends on the operator, concentrated for gravity
and near-isotropic for pointwise subsampling. The registered
prediction was (C). The outcome was (A): $\mathrm{PR}/p<1$ under every operator at
every $d$. The operator dependence appears as an ordering within (A), not as the
isotropy (C) expected.

Four corrections were made after registration and each is recorded with its
date in the committed pre-registration file. The frozen configuration named the
wrong GEBCO type-identifier codes for direct measurement and would have selected
an empty tile set; this was caught before any data was read. The evaluation set
was enlarged from $720$ to $8\,108$ tiles because $\mathrm{PR}/p$ had not
converged at the smaller size; $\mathrm{PR}/p$ rises with sample size, so the
enlargement moves the headline number against the method. Three of the $18$
cells remain unconverged at the $1\%$ tolerance at $8\,108$ tiles and are marked
in Table~\ref{tab:real-placement}. The registered prediction (C) was recorded as
refuted.

\paragraph{Data and split.} The data are the GEBCO 2024 Grid \citep{gebco2024},
which GEBCO places in the public domain; it may be used free of charge with
acknowledgement of the source, and not for navigation. Cells whose GEBCO type
identifier is not $0$, that is
anything other than a direct measurement, are removed before tiling, so no
interpolated or predicted bathymetry enters the field. Whole-block assignment and
the $32$-cell buffer mean that a basis tile and an evaluation tile are never
adjacent and never sample the same continental margin at the seam. A tile-wise
random split would leave the two sets sharing the same ridges and inflate every
agreement reported here.

\paragraph{Operators.} The gravity operator is the linearised Bouguer
attraction of the topographic mass at $14\times14=196$ stations on a plane above
the tile, with $G=6.674\times10^{-11}\,\mathrm{m^3kg^{-1}s^{-2}}$, crustal
density $2670\,\mathrm{kg\,m^{-3}}$ and a cell size of $463\,\mathrm{m}$,
reported in mGal. The gradient operator is its vertical derivative, in E\"otv\"os.
The pointwise operator samples the field at the same $196$ locations. The
noise scale is set in physical units and converted per operator so that the
signal-to-noise ratio is matched across the three.

\paragraph{Absorption.} The GEBCO fits use a single Gaussian coefficient prior,
with covariance set to the principal-component variances of the basis split, and
$\hat\sigma$ is the closed-form complement estimator
$\hat\sigma^2=\|\Pperp y\|^2/(n-r)$ averaged over realisations, so the $(n-r)$ denominator of
Proposition~\ref{prop:absorption} is already carried by the $p$ of the complement
and no finite-sample correction is made. The four arms of
Table~\ref{tab:real-absorption} differ only in which tiles supply the covariance
that defines $\rho$. The first compares a quantity with itself: it returns
$1.0000$ in all $36$ cells and departs from $1$ by a median of $0.08$ times the
sampling error of a single sample. The second, a disjoint half of every block, is
the arm Section~\ref{sec:real} quotes. Geographic separation degrades the
agreement: $\rho$ from alternate whole blocks gives $0.961$, and from the held-out
evaluation blocks $0.930$. On held-out tiles $\rho$ is computed from the second moment
of the fields about the basis-split mean, so it would include any regional mean
shift. The shift lies almost entirely in $\spn(\Phi)$: its out-of-span part carries at
most $2.2\times10^{-4}$ of the complement excess in any cell and arm, and $\rho$
computed about each subset's own mean changes no ratio by more than $10^{-4}$.

\begin{table}[h]
\centering
\caption{$\hat\sigma/(\sigma\sqrt{\rho})$ over the $36$ cells, by which tiles
define $\rho$. Agreement degrades with geographic separation between the two
sides, not with $d$.}
\label{tab:real-absorption}
\small
\begin{tabular}{lrrr}
\toprule
tiles defining $\rho$ & min & median & max \\
\midrule
same tiles & 0.9998 & 1.0000 & 1.0004 \\
disjoint half of each block & 0.9975 & 1.0090 & 1.0164 \\
alternate whole blocks & 0.9281 & 0.9615 & 1.0063 \\
evaluation split & 0.8958 & 0.9296 & 0.9990 \\
\bottomrule
\end{tabular}

\end{table}

\paragraph{Placement and the sweep.} Table~\ref{tab:real-placement} gives the
population quantities, all computed from the evaluation fields with no model
fitted. $\mathrm{PR}/p$ rises with sample size: $15$ of the $18$ cells
have moved by less than $1\%$ over the last doubling of the sample, and the
three that have not are still rising. Table~\ref{tab:real-sweep} gives all $36$
cells, with detection at $N=100$ over $20$ resamples of the evaluation set. Over
the $36$ cells at $N=100$ the mean
absolute gap between measured detection and the power law's prediction is
$0.012$. One cell departs by more than $0.15$: gravity at
$d=64$, where $\Upop=0.040$ and the normal approximation to $U$ is weakest.

\begin{table}[h]
\centering
\caption{Population placement of real topography. $\mathrm{PR}$ is the
participation ratio of the complement excess eigenvalues, $p$ its dimension,
leak the fraction of out-of-basis signal energy that falls inside
$\range(G)$, and $\Upop$ the population value of the statistic. The last column
records whether $\mathrm{PR}/p$ moved by less than $1\%$ over the final
doubling of the evaluation set.}
\label{tab:real-placement}
\small
\begin{tabular}{llrrrrrrl}
\toprule
operator & $d$ & $p$ & PR & PR$/p$ & leak & $\Upop$ & $\rho$ & conv. \\
\midrule
gravity & 4 & 192 & 6.8 & 0.036 & 0.050 & 25.002 & 25.13 & yes \\
gravity & 8 & 188 & 13.0 & 0.069 & 0.109 & 10.783 & 9.60 & yes \\
gravity & 16 & 180 & 16.2 & 0.090 & 0.109 & 5.797 & 4.13 & yes \\
gravity & 32 & 164 & 24.5 & 0.149 & 0.181 & 1.062 & 1.76 & yes \\
gravity & 64 & 132 & 33.1 & 0.251 & 0.249 & 0.040 & 1.13 & yes \\
gravity & 96 & 100 & 39.4 & 0.394 & 0.336 & 0.002 & 1.04 & yes \\
gradient & 4 & 192 & 22.9 & 0.119 & 0.100 & 7.374 & 2556.18 & yes \\
gradient & 8 & 188 & 25.3 & 0.134 & 0.051 & 6.428 & 1861.90 & yes \\
gradient & 16 & 180 & 31.6 & 0.176 & 0.071 & 4.681 & 1073.28 & yes \\
gradient & 32 & 164 & 37.4 & 0.228 & 0.101 & 3.371 & 484.47 & yes \\
gradient & 64 & 132 & 47.3 & 0.358 & 0.161 & 1.767 & 152.53 & yes \\
gradient & 96 & 100 & 49.3 & 0.493 & 0.241 & 1.000 & 68.32 & yes \\
pointwise & 4 & 192 & 20.6 & 0.107 & 0.006 & 8.175 & 124.35 & yes \\
pointwise & 8 & 188 & 32.4 & 0.172 & 0.010 & 4.694 & 83.13 & yes \\
pointwise & 16 & 180 & 48.0 & 0.267 & 0.022 & 2.643 & 52.62 & yes \\
pointwise & 32 & 164 & 64.3 & 0.392 & 0.052 & 1.453 & 31.44 & no \\
pointwise & 64 & 132 & 80.2 & 0.607 & 0.132 & 0.575 & 17.39 & no \\
pointwise & 96 & 100 & 70.9 & 0.709 & 0.239 & 0.348 & 12.62 & no \\
\bottomrule
\end{tabular}

\end{table}

\paragraph{Gradient operator.} Under the gradient operator, field coverage rises
as the reconstruction gets worse. Over the twelve gradient cells field coverage is
$0.821$--$0.954$ and field RMSE is
$10.5$ to $27.0$ times the in-span floor: coverage is $0.930$ at $d=8$ with the
error at $10.5$ times the floor, and $0.954$ at $d=96$ at $27.0$ times. The
absorbed $\hat\sigma$ inflates the posterior spread enough to cover a displaced
mean, so field coverage is uninformative here too. The sphericity test fires at
$1.000$ in all twelve gradient cells at every $N\ge25$.

\begin{table}[t]
\centering
\caption{The $d$-and-$\sigma$ sweep on GEBCO topography. $\sigma$ is the nominal
noise in mGal for the gravity arm and the matched value for the others; f.cov
and o.cov are field and observation coverage at the nominal $0.90$; RMSE/floor
is field RMSE divided by the in-span floor $\|a-P_{\spn(\Phi)}a\|$, the error an
oracle restricted to the same basis would make.}
\label{tab:real-sweep}
\small
\begin{tabular}{llrrrrrrrrrr}
\toprule
operator & $d$ & $\sigma$ & $p$ & $\hat\sigma/\sigma$ & $U$ & $\Rstd$ & $\hat\rho$ & f.cov & o.cov & RMSE/floor & det. \\
\midrule
gravity & 4 & 1 & 192 & 4.69 & 25.004 & 815.17 & 21.26 & 0.586 & 0.929 & 1.02 & 1.00 \\
gravity & 8 & 1 & 188 & 2.82 & 10.786 & 157.28 & 8.16 & 0.619 & 0.933 & 1.02 & 1.00 \\
gravity & 16 & 1 & 180 & 1.87 & 5.796 & 49.30 & 3.57 & 0.698 & 0.939 & 1.03 & 1.00 \\
gravity & 32 & 1 & 164 & 1.27 & 1.077 & 8.91 & 1.57 & 0.757 & 0.946 & 1.06 & 1.00 \\
gravity & 64 & 1 & 132 & 1.06 & 0.056 & 1.69 & 1.06 & 0.837 & 0.963 & 1.24 & 0.95 \\
gravity & 96 & 1 & 100 & 1.02 & 0.015 & 1.06 & 1.01 & 0.880 & 0.971 & 1.52 & 0.20 \\
gravity & 4 & 0.25 & 192 & 18.36 & 26.978 & 5511.39 & 138.43 & 0.583 & 0.929 & 1.02 & 1.00 \\
gravity & 8 & 0.25 & 188 & 10.60 & 13.243 & 1146.97 & 53.84 & 0.612 & 0.931 & 1.02 & 1.00 \\
gravity & 16 & 0.25 & 180 & 6.41 & 9.703 & 378.74 & 21.41 & 0.685 & 0.935 & 1.03 & 1.00 \\
gravity & 32 & 0.25 & 164 & 3.31 & 4.871 & 72.99 & 6.50 & 0.733 & 0.943 & 1.04 & 1.00 \\
gravity & 64 & 0.25 & 132 & 1.71 & 1.376 & 10.28 & 1.97 & 0.786 & 0.961 & 1.06 & 1.00 \\
gravity & 96 & 0.25 & 100 & 1.27 & 0.247 & 3.25 & 1.24 & 0.822 & 0.978 & 1.13 & 1.00 \\
gradient & 4 & 1 & 192 & 46.27 & 7.375 & 249.35 & 16.27 & 0.821 & 0.924 & 16.37 & 1.00 \\
gradient & 8 & 1 & 188 & 38.65 & 6.428 & 173.82 & 12.40 & 0.930 & 0.928 & 10.50 & 1.00 \\
gradient & 16 & 1 & 180 & 29.43 & 4.681 & 82.55 & 8.23 & 0.936 & 0.932 & 12.58 & 1.00 \\
gradient & 32 & 1 & 164 & 20.01 & 3.372 & 38.64 & 4.67 & 0.947 & 0.939 & 15.59 & 1.00 \\
gradient & 64 & 1 & 132 & 11.80 & 1.767 & 11.58 & 2.28 & 0.954 & 0.952 & 21.39 & 1.00 \\
gradient & 96 & 1 & 100 & 8.11 & 1.001 & 5.96 & 1.59 & 0.954 & 0.964 & 27.02 & 1.00 \\
gradient & 4 & 0.25 & 192 & 185.03 & 7.380 & 250.89 & 16.37 & 0.821 & 0.924 & 16.37 & 1.00 \\
gradient & 8 & 0.25 & 188 & 154.57 & 6.434 & 174.78 & 12.47 & 0.930 & 0.928 & 10.50 & 1.00 \\
gradient & 16 & 0.25 & 180 & 117.64 & 4.689 & 83.13 & 8.28 & 0.936 & 0.932 & 12.58 & 1.00 \\
gradient & 32 & 0.25 & 164 & 79.94 & 3.384 & 38.90 & 4.69 & 0.947 & 0.939 & 15.60 & 1.00 \\
gradient & 64 & 0.25 & 132 & 47.05 & 1.789 & 11.74 & 2.30 & 0.954 & 0.952 & 21.39 & 1.00 \\
gradient & 96 & 0.25 & 100 & 32.20 & 1.028 & 6.07 & 1.61 & 0.954 & 0.964 & 27.03 & 1.00 \\
pointwise & 4 & 1 & 192 & 10.37 & 8.178 & 122.84 & 6.24 & 0.647 & 0.928 & 1.00 & 1.00 \\
pointwise & 8 & 1 & 188 & 8.32 & 4.693 & 55.62 & 4.35 & 0.727 & 0.930 & 1.01 & 1.00 \\
pointwise & 16 & 1 & 180 & 6.63 & 2.643 & 21.09 & 2.94 & 0.787 & 0.933 & 1.01 & 1.00 \\
pointwise & 32 & 1 & 164 & 5.24 & 1.454 & 9.21 & 2.01 & 0.835 & 0.940 & 1.03 & 1.00 \\
pointwise & 64 & 1 & 132 & 4.08 & 0.575 & 3.75 & 1.39 & 0.875 & 0.952 & 1.09 & 1.00 \\
pointwise & 96 & 1 & 100 & 3.55 & 0.351 & 2.91 & 1.23 & 0.897 & 0.962 & 1.18 & 1.00 \\
pointwise & 4 & 0.25 & 192 & 41.30 & 8.301 & 129.26 & 6.52 & 0.647 & 0.928 & 1.00 & 1.00 \\
pointwise & 8 & 0.25 & 188 & 33.04 & 4.801 & 58.63 & 4.54 & 0.726 & 0.930 & 1.01 & 1.00 \\
pointwise & 16 & 0.25 & 180 & 26.23 & 2.740 & 22.30 & 3.05 & 0.786 & 0.933 & 1.01 & 1.00 \\
pointwise & 32 & 0.25 & 164 & 20.62 & 1.543 & 9.74 & 2.07 & 0.834 & 0.939 & 1.03 & 1.00 \\
pointwise & 64 & 0.25 & 132 & 15.85 & 0.642 & 3.99 & 1.42 & 0.874 & 0.951 & 1.08 & 1.00 \\
pointwise & 96 & 0.25 & 100 & 13.65 & 0.407 & 3.13 & 1.26 & 0.896 & 0.961 & 1.16 & 1.00 \\
\bottomrule
\end{tabular}

\end{table}

\paragraph{Spatial dependence.} Neighbouring evaluation tiles sit about $15$~km
apart, and the raw fields are strongly correlated at that range. The test reads
only complement residuals. For tiles $i\neq j$,
$\Cov(z_i,z_j)=U_\perp^\top H\,\Cov(\xi_{\perp,i},\xi_{\perp,j})\,H^\top U_\perp$
exactly, because $U_\perp^\top H\Phi=0$, so in-span content contributes nothing to
the cross-covariance under either hypothesis. Over all $18$
operator-and-$d$ cells, Table~\ref{tab:real-autocorr} separates first- from
second-moment dependence in the noiseless complement signal. The normalised
cross-covariance,
$\overline{(z_i-\bar z)^\top(z_j-\bar z)}/\overline{\|z-\bar z\|^2}$, is the
quantity the identity concerns and the one the sample covariance weights. Below
$20$~km it lies between $-0.012$ and $+0.026$, against $0.963$ for
the raw fields. The direction column is the mean cosine between centred vectors;
it is unweighted, so the many low-energy tiles dominate it. The dependence lives
in the log-energies. Tiles close together share a roughness level that no linear
projection removes: their log complement energies correlate at $0.859$ to $0.894$,
and still at $0.625$ to $0.690$ at $160$--$320$~km. This costs effective sample
size under the alternative. Under $H_0$ the complement is $U_\perp^\top\varepsilon$,
and an exact-null arm that keeps the tiles and replaces only the out-of-span part
is calibrated (Table~\ref{tab:real-thinning}).

\begin{table}[t]
\centering
\caption{Dependence between pairs of evaluation tiles by great-circle
separation. Complement columns give the range over the $18$ operator-and-$d$
cells. First moments decorrelate; log-energies do not.}
\label{tab:real-autocorr}
\small
\begin{tabular}{lrrrrrr}
\toprule
 & & \multicolumn{2}{c}{cross-covariance} & \multicolumn{2}{c}{direction} & log-energy \\
\cmidrule(lr){3-4}\cmidrule(lr){5-6}\cmidrule(lr){7-7}
separation (km) & pairs & raw & complement & raw & complement & complement \\
\midrule
0--20 & 2740 & $+0.963$ & $-0.012$ to $+0.026$ & $+0.871$ & $+0.033$ to $+0.117$ & $+0.859$ to $+0.894$ \\
20--40 & 6653 & $+0.854$ & $-0.019$ to $+0.003$ & $+0.836$ & $+0.031$ to $+0.099$ & $+0.770$ to $+0.801$ \\
40--80 & 23200 & $+0.787$ & $-0.003$ to $+0.005$ & $+0.797$ & $+0.034$ to $+0.101$ & $+0.706$ to $+0.750$ \\
80--160 & 53849 & $+0.686$ & $-0.001$ to $+0.002$ & $+0.749$ & $+0.032$ to $+0.089$ & $+0.632$ to $+0.673$ \\
160--320 & 59707 & $+0.511$ & $-0.003$ to $+0.002$ & $+0.662$ & $+0.038$ to $+0.097$ & $+0.625$ to $+0.690$ \\
320--1000 & 69308 & $+0.372$ & $-0.001$ to $+0.001$ & $+0.496$ & $+0.023$ to $+0.049$ & $+0.045$ to $+0.104$ \\
1000--20000 & 2908293 & $-0.041$ & $-0.000$ to $+0.000$ & $+0.051$ & $+0.012$ to $+0.026$ & $-0.038$ to $-0.034$ \\
\bottomrule
\end{tabular}

\end{table}

Table~\ref{tab:real-thinning} thins the evaluation set to a minimum separation
and pairs each thinned subset with a random subset of the same size, which
separates dependence from sample size. Thinning to a minimum
separation of $120$~km discards $98.5\%$ of the tiles and leaves detection at
$1.000$ in four of the five cells thinned; those four are
saturated at every level. In the informative cell, gravity at $d=96$, thinning
detects better than a size-matched subset where the comparison discriminates,
$0.63$ against $0.47$ at $30$~km and $0.38$ against $0.28$ at $60$~km; at the two
levels where it is lower, the difference is $0.05$, about one standard error at
$60$ trials. At one tile per block, $N=23$ falls below the $N>p/2$ condition, so
$\Rstd$ is undefined there, while the test on $U$ is unaffected.

\begin{table}[h]
\centering
\caption{Detection under spatial thinning, as thinned$/$matched-$N$. Each
column is one (operator, $d$) cell. The exact-null false-positive rate over
these six levels and the three complement dimensions has median $0.040$ against
a nominal $0.05$.}
\label{tab:real-thinning}
\small
\begin{tabular}{lrrrrrr}
\toprule
thinning & $N$ & grav~32 & grav~64 & grav~96 & grad~96 & poin~96 \\
\midrule
$\ge 0$ km & 8108 & 1.00/1.00 & 1.00/1.00 & 1.00/1.00 & 1.00/1.00 & 1.00/1.00 \\
$\ge 15$ km & 3035 & 1.00/1.00 & 1.00/1.00 & 1.00/0.95 & 1.00/1.00 & 1.00/1.00 \\
$\ge 30$ km & 1065 & 1.00/1.00 & 1.00/1.00 & 0.63/0.47 & 1.00/1.00 & 1.00/1.00 \\
$\ge 60$ km & 336 & 1.00/1.00 & 1.00/1.00 & 0.38/0.28 & 1.00/1.00 & 1.00/1.00 \\
$\ge 120$ km & 119 & 1.00/1.00 & 1.00/0.98 & 0.12/0.17 & 1.00/1.00 & 1.00/1.00 \\
one per block & 23 & 1.00/0.98 & 0.53/0.65 & 0.10/0.15 & 1.00/1.00 & 1.00/1.00 \\
\bottomrule
\end{tabular}

\end{table}

\paragraph{Competitors.} Table~\ref{tab:real-competitors} counts how often each
diagnostic fires over the $36$ cells. The four unharmed cells still have
$\rho>1$, so the right-hand column is not a false-positive rate and specificity
is not estimable on these data. Cross-validated predictive coverage and held-out
observation coverage both miss the entire gradient arm below $d=96$.

\begin{table}[h]
\centering
\caption{What fires, on real topography. The per-index check is shown after
null calibration; before it, the largest of $p$ binomial coverages is $0.10$
from nominal on every draw under correct specification, and the uncalibrated
version fires in all $36$ cells.}
\label{tab:real-competitors}
\small
\begin{tabular}{lrr}
\toprule
diagnostic & fires when harmed & fires when not \\
\midrule
held-out observation coverage & 5/32 & 2/4 \\
$K$-fold cross-validation & 8/32 & 2/4 \\
per-index calibration, null-calibrated & 0/32 & 0/4 \\
full-covariance goodness of fit & 32/32 & 4/4 \\
sphericity test & 31/32 & 4/4 \\
\bottomrule
\end{tabular}

\end{table}

The full-covariance goodness-of-fit test is the competitor closest to ours.
Fixing $\Sigma_y=GSG^\top+\hat\sigma^2
I$ from the basis split and testing held-out observations against it is
correctly calibrated, with $0.055$ rejection at a nominal $0.05$ over $400$ draws
from the fitted model, and it fires in every one of the $36$ cells. It is not
scale-free: drawing from a covariance with the correct shape and the field
variance inflated by $15\%$ makes it reject on $99\%$ of draws. The
basis-to-evaluation gap in the excess is $15.7\%$ in median across the $36$
cells, so a harmless scale error of exactly that size is present everywhere, and
the test's $36/36$ does not show that it found the misspecification. In the
complement the test is exactly $\|Z\|_F^2/\hat\sigma^2_{\mathrm{ref}}$.
Referred to a $\hat\sigma$ fitted on the same sample it fires in no
cell, as Proposition~\ref{prop:forced}(ii) requires. Referred to a disjoint half of
the evaluation split it still fires at a median rate of
$0.94$. The dispersion of per-tile energy alone predicts that rate to a
median of $0.017$, against $0.55$ from the covariance anisotropy.

\paragraph{Range-block bound.} The bound
$\sigma^2\le\lambda_{\min}(\Cov(U_\parallel^\top y))$ has slack at least $1.04$ and
median $8.8$ over the $36$ cells, so it never binds. For the single Gaussian used
here the fitted range covariance is the sample one, so this is also the
per-component condition of Proposition~\ref{prop:absorption}; because $\hat\sigma$ is
the closed-form estimator, the bound says when a free single-Gaussian fit would
return the same $\hat\sigma$. Its minimum, at $d=96$, shows
that the dominance condition of Proposition~\ref{prop:absorption} is only marginal
there.

\section{Experimental Setup}\label{app:setup}

This appendix specifies the synthetic benchmarks and training behind
Section~\ref{sec:experiments}. All training uses the closed-form E-step available
for linear operators. A Metropolis-adjusted Langevin E-step was $40$--$55\times$
slower for results identical to three decimals. All experiments ran on one
workstation with an AMD Ryzen 9 7950X CPU and an NVIDIA GeForce RTX 4090 GPU, in
double precision. For $N$ realisations, complement dimension $p$ and $B$ bootstrap
replicates, the test costs $O(Nnp)$ for the projection, $O(Np^2)$ for the sample
covariance and $O(p^3)$ for its eigendecomposition, and the bootstrap repeats the
last two $B$ times, for $O(Nnp+B(Np^2+p^3))$ time and $O(Np+p^2)$ memory.

\begin{table}[t]
\caption{Benchmarks, ordered by misspecification ratio. Two are in-basis.
Fields are generated as $a=\Phi z+\xi$ with $\xi$ stationary Mat\'ern-1/2, so
out-of-basis content is present by construction on the remaining seven.}
\label{tab:benchmarks}

\centering
\begin{tabular}{lllrrr}
\toprule
Problem & operator & $d/K$ & $n_{\text{obs}}$ & $\sigma$ & $\rho$ \\
\midrule
\texttt{syn\_inbasis} & integral   & 4/2  & 100 & 0.01  & 1.00 \\
Darcy (2D)            & 2D PDE     & 9/2  & 144 & 0.005 & 1.00 \\
Low-SNR               & integral   & 4/2  & 100 & 0.1   & 1.09 \\
Darcy (1D)            & 1D PDE     & 4/2  & 100 & 0.01  & 3.01 \\
High-dim              & integral   & 12/3 & 100 & 0.01  & 4.32 \\
Bimodal               & integral   & 4/2  & 100 & 0.01  & 9.54 \\
Trimodal              & integral   & 4/3  & 100 & 0.01  & 9.54 \\
Sensors               & 20 sensors & 4/2  & 20  & 0.01  & 16.15 \\
Deconv                & smoothing  & 5/2  & 100 & 0.01  & 195.26 \\
\bottomrule
\end{tabular}

\end{table}

\paragraph{Generation.}
Every benchmark is produced by the shipped generator, whose fields lie at a
distance of $1.4$--$3.6\times10^{-3}$ from $\spn(\Phi)$. Six of seven benchmarks
other than Trimodal regenerate bit-exactly. Trimodal shares all generator settings
with Bimodal except $K$ and component separation. Its field RMSE is $0.3672$ and
its $\hat\sigma/\sigma$ is $2.98$, and it is statistically indistinguishable from
Bimodal.

\paragraph{Baselines.}
\textsc{gpr-spectral} is GP regression with the basis-restricted kernel under
test (rank $d$), so it is numerically the same estimator as the model.
\textsc{gpr-fullrank} uses a stationary Mat\'ern kernel at full rank.

\begin{figure}[t]
\centering
\includegraphics[width=\textwidth]{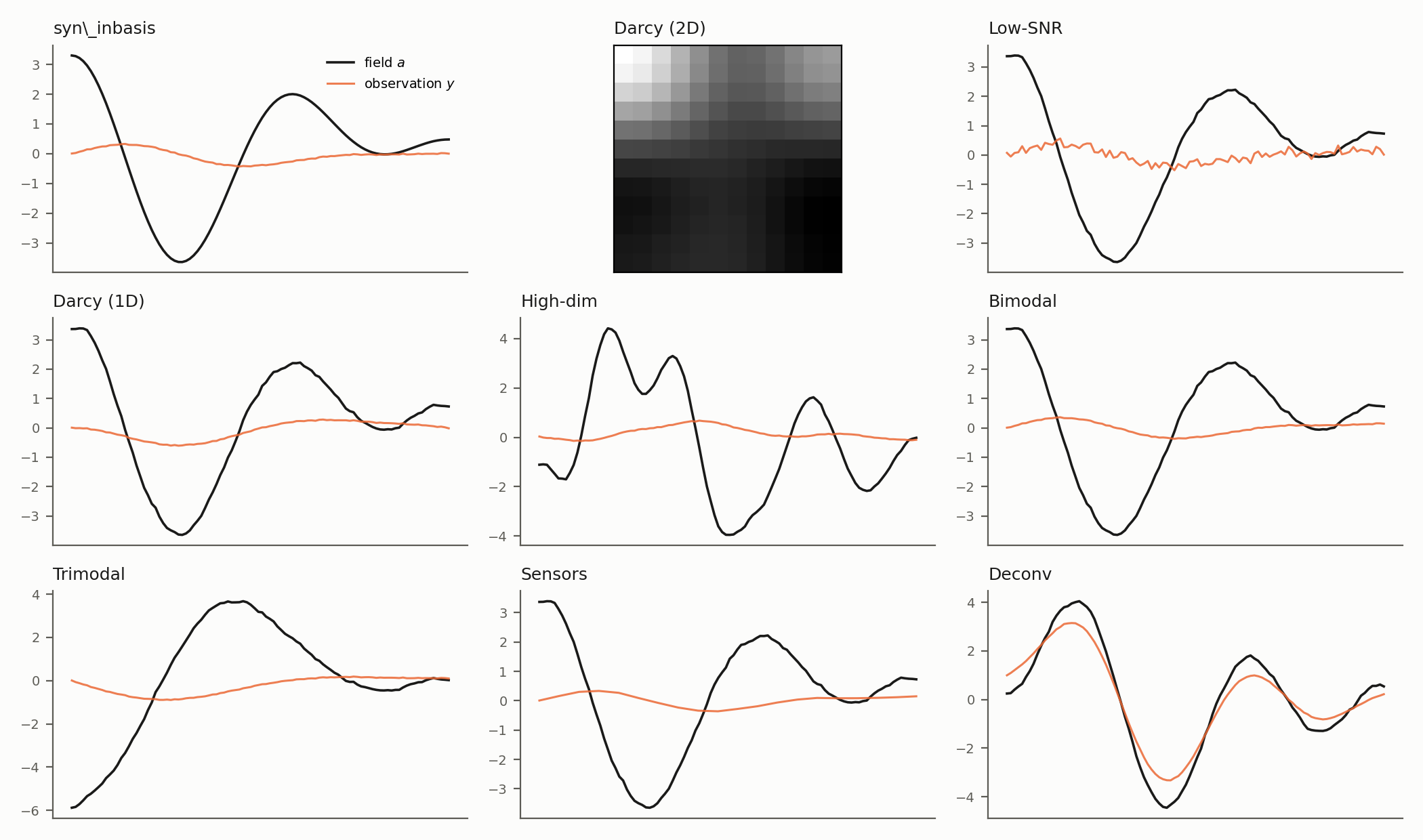}
\caption{The nine benchmarks: the true field $a$ and the observation $y$ it
produces under each operator. Darcy (2D) is shown as an image; the rest are 1D.}
\label{fig:suite}
\end{figure}

\section{Scoping and the Constrained Prior}\label{app:scoping}

This appendix supports the Discussion's statement that the restriction helps
in-basis and hurts out of basis, and reports a constrained prior as a candidate
remedy.

\begin{table}[t]
\centering
\small
\begin{tabular}{l|cccc|ccc}
\toprule
& \multicolumn{4}{c|}{field RMSE} & \multicolumn{3}{c}{field cov$_{90}$}\\
Problem & ours & \textsc{gpr-spec} & \textsc{gpr-full} & in-span floor & ours & \textsc{gpr-spec} & \textsc{gpr-full}\\
\midrule
\texttt{syn\_inbasis} & 0.0190 & \textbf{0.0189} & 0.1547 & 0.0028 & 0.859 & 0.870 & 0.994\\
Darcy (2D) & \textbf{0.1954} & 0.2187 & 0.7717 & 0.0001 & 0.874 & 0.874 & 0.922\\
Low-SNR & 0.4614 & 0.4658 & \textbf{0.4554} & 0.1618 & 0.662 & 0.657 & 0.985\\
Darcy (1D) & 0.4449 & 0.4448 & \textbf{0.3153} & 0.1618 & 0.180 & 0.187 & 0.963\\
High-dim & 0.6889 & 0.6768 & \textbf{0.3686} & 0.1377 & 0.768 & 0.795 & 0.962\\
Bimodal & 0.4419 & 0.4404 & \textbf{0.1583} & 0.1618 & 0.377 & 0.364 & 0.993\\
Trimodal & 0.4419 & 0.4405 & \textbf{0.1499} & 0.1618 & 0.377 & 0.366 & 0.994\\
Sensors & 0.4363 & 0.4407 & \textbf{0.1771} & 0.1618 & 0.693 & 0.592 & 0.995\\
Deconv & 0.1529 & 0.1527 & \textbf{0.0787} & 0.1526 & 0.318 & 0.093 & 1.000\\
\bottomrule
\end{tabular}

\caption{Field RMSE and coverage of the rank-$d$ prior and the two GP
baselines, ordered by $\rho$ as in Table~\ref{tab:benchmarks}. The in-span floor
is $\|a-P_{\spn(\Phi)}a\|_{\mathrm{rms}}$, the out-of-span residual norm alone.}
\label{tab:scoping}
\end{table}

\begin{figure}[t]
\centering
\includegraphics[width=\textwidth]{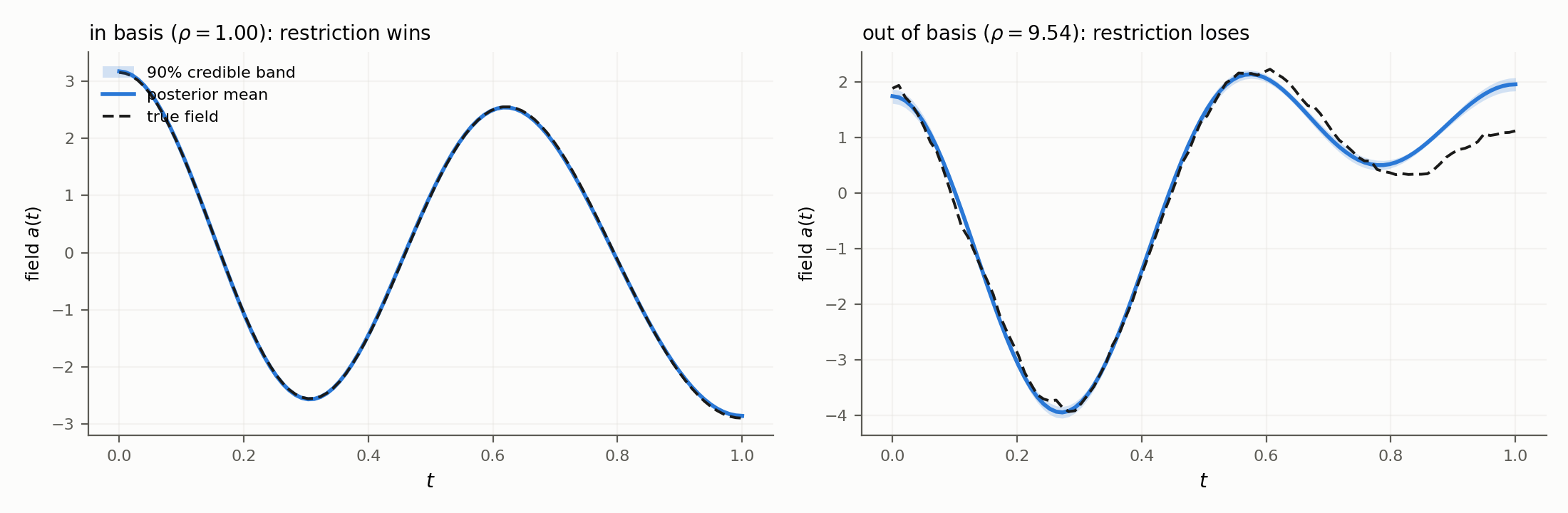}
\caption{The rank-$d$ prior on \texttt{syn\_inbasis} ($\rho=1.00$) and Bimodal
($\rho=9.54$): same operator family, model and nominal level. Each panel shows the
test case at the median field RMSE among the 100 held-out cases for its
configuration.}
\label{fig:invsout}
\end{figure}

\paragraph{Correspondence with $\rho$.}
On our nine benchmarks, the restriction wins on exactly the two problems with
$\rho\approx 1$. They span
two operator families, a 1D integral operator and a 2D elliptic PDE, so the effect
comes from being in-basis rather than from a particular operator. The margins are
$8.2\times$ on \texttt{syn\_inbasis} and $3.9\times$ on Darcy (2D). On the seven
out-of-basis problems the full-rank prior wins by $1.02\times$ to $2.95\times$.
What transfers is the sign of the effect, and the fact that a statistic
computable without ground truth separates the two regimes. The suite spans nine
problems, so it does not locate the boundary in general, and we propose no
threshold on $\rho$.

\paragraph{Coverage of the two priors.}
The full-rank prior covers at $0.92$--$1.00$ throughout, that is, conservatively;
the rank-$d$ prior covers at $0.18$--$0.87$. This is the oversmoothing/rougher
dichotomy of \citet{knapik2011bayesian}. It is visible even on
\texttt{syn\_inbasis}, where the rank-$d$ prior is
correct ($0.859$) and the full-rank prior over-covers ($0.994$).

\paragraph{Composition of the excess over the floor.}
For any estimator confined to $\spn(\Phi)$,
$\|a-\hat a\|^2 = \text{floor}^2 + \|P_\Phi a-\hat a\|^2$ exactly, case by case. We
define the in-span excess as the root mean square of $\|P_\Phi a-\hat a\|$ over the
held-out cases. Because field RMSE and the floor are root mean squares over the
same cases, the excess equals $\sqrt{\text{RMSE}^2-\text{floor}^2}$ for each seed,
and we report its mean over seeds. On Bimodal it is $0.411$ for our model and
$0.410$ for \textsc{gpr-spectral}. \textsc{gpr-spectral} has no latent variable,
so the shared excess is not latent-estimation error. The basis restriction
contributes only the $0.14$--$0.16$ floor; the rest is failure to estimate the
component \emph{inside} the basis, where both estimators fail identically. This is the $\sigma_\parallel^2$ term of
Proposition~\ref{prop:coverage}, measured directly: out-of-span field content
passes through $H$ and lands inside $\range(G)$ as correlated error.

\subsection*{The constrained prior}\label{app:constrained}

A candidate remedy enforces the operator-induced linear constraint at rank $M-c$
instead of collapsing to rank $d$ \citep{jidling2017linearly}; it is the
constrained-prior curve at rank $M-1$ in Figure~\ref{fig:sweep}(b). We use the
integral constraint $\int_0^1\phi_j=0$ for the 1D problems and a Dirichlet
boundary constraint for Darcy (2D), over 31 sweep points and 8 named benchmarks,
10 seeds and both $\sigma$ arms (780 runs).

\begin{table}[t]
\centering
\small
\begin{tabular}{lrrrrr}
\toprule
Problem & RMSE constr. & RMSE \textsc{gpr-full} & ratio & cov$_{90}$ constr. & cov$_{90}$ rank-$d$ \\
\midrule
Darcy (2D) & 0.1950 & 0.7717 & \textbf{0.25} & 0.881 & 0.874 \\
Low-SNR & 0.4630 & 0.4554 & 1.02 & 0.659 & 0.662 \\
Deconv & 0.1490 & 0.0787 & 1.89 & 0.596 & 0.318 \\
High-dim & 0.6984 & 0.3686 & 1.89 & 0.765 & 0.768 \\
Sensors & 0.4372 & 0.1771 & 2.47 & 0.744 & 0.693 \\
Darcy (1D) & 0.8055 & 0.3153 & 2.55 & 0.978 & 0.180 \\
Bimodal & 0.9684 & 0.1583 & 6.12 & 0.965 & 0.377 \\
Trimodal & 0.9675 & 0.1499 & 6.46 & 0.965 & 0.377 \\
\bottomrule
\end{tabular}

\caption{Constrained prior at rank $M-1$, trained to convergence. Constraint
satisfied to $5.2\times10^{-15}$ over all 780 runs.}
\label{tab:constrained}
\end{table}

The constrained prior over-covers rather than under-covering
(Table~\ref{tab:constrained}), and pinning $\sigma$ at truth leaves this intact
(field coverage $0.948$ against $0.965$ with $\sigma$ free on Bimodal), so it is
the rougher-prior half of the classical dichotomy \citep{knapik2011bayesian},
not absorption. It does not restore accuracy: field RMSE is $1.9$--$6.5\times$
the full-rank baseline on six of eight benchmarks, although the constrained
kernel has rank $99$ against the baseline's $100$. Its coverage is conservative
rather than correct, so we recommend it only as the safe failure direction, with
the physical constraint satisfied exactly.

\section{Negative Results}\label{app:negative}

This appendix shows that the coefficient prior adds little beyond the restriction
(Section~\ref{sec:setting}) and gives the in-span floor cited in the Discussion.
Our model and \textsc{gpr-spectral} agree to
within $0.003$ field RMSE on five of nine benchmarks, and to three significant
figures on \texttt{syn\_inbasis} ($0.0190$ versus $0.0189$;
Table~\ref{tab:scoping}). The wins of Appendix~\ref{app:scoping} therefore belong
to the \emph{rank-$d$ basis restriction}, which plain GP regression with the same
kernel achieves equally. Across all nine benchmarks the latent apparatus buys a
measurable improvement only on Darcy (2D), $0.196$ against $0.220$, about $11\%$.
The results are therefore stated for a general basis-restricted Gaussian prior.
The coefficient prior affects what the model reports about latent structure, but
not reconstruction, coverage or either diagnostic.

\paragraph{The in-span floor.} The in-span floor
$\|a-P_{\spn(\Phi)}a\|_{\mathrm{rms}}$ is the out-of-span residual alone, the best any
basis-restricted estimator could achieve. On the Bimodal benchmark that floor is $0.162$,
and the in-span excess over it (Appendix~\ref{app:scoping}) is $0.411$ for our model
and $0.410$ for Gaussian process regression with the same kernel. That excess lies inside $\spn(\Phi)$, where no statistic computed on the complement
can reach it. The floor is defined for basis-restricted estimators only, and
\textsc{gpr-fullrank} reaches $0.158$---below the
floor of $0.162$---because it is not confined to $\spn(\Phi)$. The gap is therefore a cost of
the restriction rather than of either fit: the two restricted estimators pay it identically
and the unrestricted one does not pay it.

\section{Status of Every Result}\label{app:status}

{\footnotesize
\setlength{\LTcapwidth}{\textwidth}
\begin{longtable}{>{\raggedright\arraybackslash}p{0.30\textwidth}%
                  >{\raggedright\arraybackslash}p{0.17\textwidth}%
                  >{\raggedright\arraybackslash}p{0.44\textwidth}}
\caption{Every result in the paper with its status, extending Table~\ref{tab:status}. \textsc{exact}: holds at every finite
$(N,p)$, no asymptotics, no unstated condition. \textsc{approximation}: exact in a regime we
name, approximate outside it, departure measured. \textsc{asymptotic}: holds in a stated
limit and says nothing at finite $(N,p)$---it never stands in for an exact row.
\textsc{empirical}: measured over a grid we state, not extrapolated beyond it.
\textsc{measured}: a property of our own suite, reported as evidence of an effect's size
rather than as a law.}\label{tab:status-full}\\
\toprule
result & status & condition it rests on \\
\midrule
\endfirsthead
\multicolumn{3}{l}{\emph{Table~\ref{tab:status-full}, continued}}\\
\toprule
result & status & condition it rests on \\
\midrule
\endhead
\midrule
\multicolumn{3}{r}{\emph{continued on the next page}}\\
\endfoot
\bottomrule
\endlastfoot
\multicolumn{3}{l}{\emph{Exact at every finite $(N,p)$}}\\*
Prop.~\ref{prop:sphericity}(a), scale invariance & \textsc{exact} & none; algebraic \\
Prop.~\ref{prop:sphericity}(b), the null model $\Cov(z)=\sigma^2I_p$ & \textsc{exact} & Gaussian $\varepsilon$; $\xi_\perp=0$ and the model mean in $\range(G)$ \\
Calibration of that null & \textsc{empirical} & parametric Monte-Carlo bootstrap, $1000$ replicates at each $(N,p)$; not exact \\
Prop.~\ref{prop:sphericity}(c), $\lambda_1=\sigma^2(1+(\rho-1)p/k)$ & \textsc{exact} & $A$ has $k$ equal nonzero eigenvalues \\
Prop.~\ref{prop:forced}, given $\|Z\|_F^2$ every exact test on the complement is a test of direction & \textsc{exact} & Gaussian $\varepsilon$; completeness and Neyman structure \citep{lehmann2005testing} \\
Prop.~\ref{prop:sphericity}(d), no power under isotropy & \textsc{exact} & Gaussian $\xi_\perp$ and $A=(\rho-1)I_p$; a statement about $\operatorname{median}(\Rstd\mid H_0)$, not its mean \\
Likelihood separation $\ell=\ell_\parallel+\ell_\perp$ (Prop.~\ref{prop:absorption}) & \textsc{exact} & linear $H$, basis-restricted covariance \\
Lem.~\ref{lem:nullmean}, $\E[U\mid H_0]=\frac{(p-1)(p+2)}{\nu p+2}$ & \textsc{exact} & Gaussian $\varepsilon$; every $(\nu,p)$, $p>\nu$ included; median $0.041\%$ over $59$ pairs \\
\midrule
\multicolumn{3}{l}{\emph{True in a limit, not at finite $(N,p)$}}\\*
Prop.~\ref{prop:consistency}, power $\to1$ & \textsc{asymptotic} & $N\to\infty$ at fixed $p$; $\Sigma\not\propto I_p$ \\
Null law $\nu U-p\to\Nor(\E[x^4]-2,4)$, $\E[x^4]$ the fourth moment of a standardised entry & \textsc{asymptotic}, cited & any $p/\nu\to[0,\infty]$; \citet{ledoitwolf2002,wangyao2013,liyao2016} \\
Power $\Phi\big(\tfrac{(\nu+1)U_{\mathrm{pop}}-2z_\alpha}{2(1+U_{\mathrm{pop}})}\big)$ & \textsc{asymptotic}, cited & \citet{wangyao2013} ($p/\nu\to c$, few spikes), \citet{liyao2016} ($p/\nu\to\infty$, general $\Sigma$) \\
Thm.~\ref{thm:prop}, $U$ at proportional $p/\nu$, and Cor.~\ref{cor:proppower} & \textsc{asymptotic} & Gaussian data; $p/\nu\to y\in(0,\infty)$; $\sup_{p}\,\|T_p\|<\infty$; via \citet[Thm.~9.10]{baisilverstein2010} \\
Law \eqref{eq:proplaw} and its third-order refinement at finite $(\nu,p)$ & \textsc{empirical} & our $949$ Gaussian cells: mean abs.\ error $0.016$ and $0.011$, against $0.019$ for \eqref{eq:ly} \\
\midrule
\multicolumn{3}{l}{\emph{Approximations valid in a stated regime}}\\*
Prop.~\ref{prop:absorption}, $\hat\sigma/\sigma=\sqrt\rho$ & \textsc{approximation} & in-span prior scale dominates the noise; implemented update normalises by $n$; matched to $3.4\%$, or $0.5\%$ with the $\sqrt{(n-r)/n}$ factor (Rem.~\ref{rem:breakdown}) \\
Prop.~\ref{prop:coverage}, two-term coverage & \textsc{approximation} & likelihood-dominated, $\sigma$ small, $n\gg d$; up to a constant $1.22$ on $\hat v/\Var(e)$ (Rem.~\ref{rem:offset}) \\
Cor.~\ref{cor:identity}, the free/pinned identity & \textsc{approximation} & inherits the above; within $9\%$ for $\rho\lesssim10$, $28\%$ at $\rho=306$ (Rem.~\ref{rem:regime}) \\
\midrule
\multicolumn{3}{l}{\emph{Empirical regularities, measured over the grid stated in Appendix~\ref{app:power}}}\\*
$\Rstd$ at detection $0.8$: a band, $1.17$ ($N>2p$) to $1.26$ ($p/2<N\le p$), median $1.19$ & \textsc{empirical} & $N>p/2$, $p\in[16,250]$; CV $0.115$ over $21$ series; \textbf{not a universal cutoff};\ a consequence of the power law (\S\ref{app:asymptotics}) \\
$\hat\rho$ accurate to $2.1\%$ over $\rho\in[1,5]$ & \textsc{empirical} & our benchmarks, where $k\approx1$; needs $N>p$ and $k<p/2$ \\
$\hat\rho-1$ under-stated $\rho-1$ in every cell where the test fired & \textsc{empirical} & the $527$ simulated cells with detection $\ge0.8$; median error $-19.1\%$ \\
Blind-spot width, $k/p$ at detection $0.8$ & \textsc{empirical} & $p=96$; $0.83$ at $N=800$, $0.50$ at $N=100$ \\
Clause (d) at the nominal rate & \textsc{empirical} & $98$ isotropic cells; detection $0.045$, $\Rstd=1.0018$ \\
False-positive rate $0.0509$ & \textsc{empirical} & $98$ null cells, $59$ $(N,p)$ pairs, $19\,600$ draws \\
$1.22$ constant on $\hat v/\Var(e)$ & \textsc{empirical} & free-$\sigma$ arm, $\rho\ge2$; CV $0.016$ \\
\midrule
\multicolumn{3}{l}{\emph{Controlled measurements on our own suite, not general laws}}\\*
Observation coverage $0.878$--$0.897$ while field coverage falls $0.851\to0.331$ & \textsc{measured} & 31 sweep points, 10 seeds, one model family \\
Basis restriction helps/hurts by $8.2\times$/$3.0\times$ & \textsc{measured} & our nine benchmarks \\
In-span excess over the floor, $0.411$ and $0.410$ on Bimodal & \textsc{measured} & Bimodal; our model and \textsc{gpr-spectral} \\
Coefficient prior adds little beyond the restriction; in-span excess unreachable by the complement test; constrained prior restores conservatism, not accuracy (Appendices~\ref{app:scoping}, \ref{app:negative}) & \textsc{measured} & as stated in each section \\
\end{longtable}}

\end{document}